%% file: Secrecy_Energy_Efficiency.tex
\documentclass[onecolumn, draftcls, 12pt]{IEEEtran}

\let\labelindent\relax
\usepackage{cite}

\usepackage{graphics}

\usepackage[cmex10]{amsmath}
\usepackage{amsfonts}
\usepackage{amsthm}

\usepackage{algorithm}
\usepackage{algpseudocode}

\input{peter_header}

\usepackage[utf8]{inputenc}
\usepackage[english]{babel}

\newtheorem{theorem}{Theorem}[section]

\newtheorem{lemma}[theorem]{Lemma}
\newtheorem{remark}{Remark}[section]

\usepackage{bm}

\newcommand{\ma}[1]{\mathbf{ #1 }}         

\newcommand{\compl}{\mathbb{C}}        
\newcommand{\real}{\mathbb{R}}         

\newcommand{\overbar}[1]{\mkern 0.5mu\overline{\mkern-0.5mu#1\mkern-0.5mu}\mkern 0.5mu}

\newcounter{mytempeqncnt1}

\newcommand{\MinorRR}[1]{{\color[rgb]{0.0,0.0,0}#1}}

\newcommand{\MinorRevOmid}[1]{{\color[rgb]{0.0,0.0,0}#1}}

\newcommand{\revOmid}[1]{{\color[rgb]{0.0,0.0,0}#1}}
\newcommand{\revPeter}[1]{{\color[rgb]{0.0,0.0,0.0}#1}}
\newcommand{\revPeterNew}[1]{{\color[rgb]{0.0,0.0,0}#1}}
\newcommand{\revPeterNewNew}[1]{{\color[rgb]{0.0,0,0}#1}}
\usepackage{enumitem}

\begin{document}
\title{Secrecy Energy Efficiency of MIMOME Wiretap Channels with Full-Duplex Jamming}
		\author{Omid Taghizadeh,~Peter~Neuhaus,~Rudolf~Mathar,~\IEEEmembership{Senior~Member,~IEEE}~and~Gerhard~Fettweis,~\IEEEmembership{Fellow,~IEEE}
			\IEEEcompsocitemizethanks{
							\IEEEcompsocthanksitem{O.~Taghizadeh and R.~Mathar are with the Institute for Theoretical Information Technology, RWTH Aachen University, Aachen, 52074, Germany (e-mail: \{taghizadeh, mathar\}@ti.rwth-aachen.de).}
							\IEEEcompsocthanksitem{\revOmid{P.~Neuhaus and G.~Fettweis are with the Vodafone Chair Mobile Communications Systems, Technische Universit\"at Dresden, Dresden 01062, Germany (e-mail: \{peter.neuhaus, gerhard.fettweis\}@ifn.et.tu-dresden.de).}}
							\IEEEcompsocthanksitem{\revOmid{Part of this work has been presented at the 2018 IEEE Wireless Communications and Networking Conference (WCNC'18), Barcelona, Spain~\cite{secrecy_WCNC_taghizadeh}.}}} } 
              
\maketitle
 \vspace{-20mm}
\begin{abstract} 
Full-duplex (FD) jamming transceivers are recently shown to enhance the information security of wireless communication systems by simultaneously transmitting artificial noise (AN) while receiving information. In this work, we investigate if FD jamming can also improve the \revPeterNew{system's} secrecy energy efficiency (SEE) in terms of securely communicated bits-per-Joule, when considering the additional power used for jamming and self-interference (SI) cancellation. \revOmid{Moreover, the degrading effect of the residual SI is also taken into account.} In this regard, \revOmid{we formulate} a set of SEE maximization problems for \revPeterNew{a} FD multiple-input-multiple-output multiple-antenna eavesdropper (MIMOME) wiretap channel, \revOmid{considering both} cases where exact or statistical channel state information (CSI) is available. Due to the intractable problem structure, \revOmid{we propose iterative solutions in each case with a proven convergence to a stationary point.} Numerical simulations indicate only a marginal SEE gain, \revPeterNew{through} the utilization of FD jamming, for a wide range of system conditions. However, \revOmid{when SI can efficiently be mitigated, the observed gain is considerable} for scenarios with a small distance between the FD node and the eavesdropper, a high signal-to-noise ratio (SNR), or for a bidirectional FD communication setup.
\end{abstract}


\IEEEpeerreviewmaketitle

\section{Introduction} \label{sec:into}
\input{./Sections/sec_into}
\vspace{-6mm}
\section{System Model}\label{sec:model}
\input{./Sections/sec_systemmodel}
\vspace{-6mm}
\section{Secrecy Energy Efficiency Maximization} \label{sec_SEE_max}
\input{./Sections/sec_SEE_max} \vspace{-6mm}
\section{Secure Bidirectional Communication: Joint Full-Duplex Operation at Alice and Bob} \label{sec_SEE_max_BD}
\input{./Sections/sec_SEE_max_BD} \vspace{-6mm}

\section{Secrecy Energy Efficiency Maximization with Statistical CSI} \label{sec_SSUM}
\input{./Sections/sec_SSUM} \vspace{-6mm}
\section {Simulation Results}\label{sec:simulations} 

\input{./Sections/sec_simulations} \vspace{-6mm}
\section {Conclusion} \label{sec:conclusion}
\input{./Sections/sec_conclusion} \vspace{-6mm}

\vspace{-6mm}
\appendices

\input{./Sections/Appendix}


%



\end{document}

%% file: peter_header.tex
\usepackage[svgnames]{xcolor} 

\usepackage{graphicx}
\usepackage{color, psfrag}
\usepackage{subfigure}

\usepackage{tkz-euclide,subfigure}
\usepackage{pstricks,pst-plot,pst-node}
\usepackage{tikz}
\usetikzlibrary{arrows,matrix,positioning}
\usepackage[section]{placeins}
\usepackage{pgfplots}
\usepgfplotslibrary{groupplots}
\usepgfplotslibrary{external}

\usepackage{hyperref}
\hypersetup{colorlinks=false, pdfborder={0 0 0}}

\usepackage[noabbrev]{cleveref}
\crefformat{equation}{#2(#1)#3}
\Crefformat{equation}{#2(#1)#3}

\definecolor{rosso}{RGB}{220,57,18}
\definecolor{giallo}{RGB}{255,153,0}
\definecolor{blu}{RGB}{102,140,217}
\definecolor{verde}{RGB}{16,150,24}
\definecolor{viola}{RGB}{153,0,153}

\pgfplotsset
    {
every axis plot/.append style={line width=1pt},
myStandard/.style=
          {
          		compat=newest,
    					trim axis left,
    					plot coordinates/math parser=false,
    					every axis/.append style={
        					xlabel shift=0.5em,
        					ylabel shift=-0.5em
    					},
          			   scale only axis,
                       title={},
                       grid=major,
                       xlabel near ticks,
                       ylabel near ticks,
                       cycle list={%
											              {blu,mark=pentagon,mark options=solid},
                                    {rosso,mark=square,mark options=solid},
																		{verde,mark=triangle,mark options=solid},																														
                                    {viola,mark=diamond,mark options=solid},
                                    {giallo,mark=o,mark options=solid},																		
                                    {red,mark=x,mark options=solid},
                                    {green,mark=+,mark options=solid},
                                    {black,mark=oplus,mark options=solid},											
					    },
                       scaled y ticks = true, tick label style={/pgf/number format/fixed},
                       }
       }
   
\pgfplotsset
{
	every axis plot/.append style={line width=1pt},
	myStandard2/.style=
	{
		compat=newest,
		trim axis left,
		plot coordinates/math parser=false,
		every axis/.append style={
			xlabel shift=0.5em,
			ylabel shift=-0.5em
		},
		scale only axis,
		title={},
		grid=major,
		xlabel near ticks,
		ylabel near ticks,
		cycle list={%
			{blu,mark=pentagon,mark options=solid},
			{rosso,mark=square,mark options=solid},
			{blu,mark=triangle,mark options=solid},
			{rosso,mark=diamond,mark options=solid},										
		},
		scaled y ticks = true, tick label style={/pgf/number format/fixed},
	}
}
       
\pgfplotsset
    {
every axis plot/.append style={line width=1pt},
myStandardCDF/.style=
          {
          			    compat=newest,
    					trim axis left,
    					plot coordinates/math parser=false,
    					every axis/.append style={
        					xlabel shift=0.5em,
        					ylabel shift=-0.5em
    					},
          			   scale only axis,
                       title={},
                       grid=major,
                       xlabel near ticks,
                       ylabel near ticks,
                       cycle list={
											              {blu},
                                    {blu,dashed},
                                    {rosso},
                                    {rosso,dashed},
                                    {verde},
                                    {verde,dashed},																		
                                    {viola},
                                    {viola,dashed},
                                    {red},
                                    {red,dashed},
                                    {green},
                                    {green,dashed},
                                    {giallo},
                                    {giallo,dashed},	
																		{black},
                                    {black,dashed}
					    },
                       scaled y ticks = true, tick label style={/pgf/number format/fixed},
                       }
       }


%% file: Sections/sec_into.tex

The information security of wireless communication systems is currently addressed via cryptographic approaches, at the upper layers of the protocol stack. However, these approaches are prone to attack due to the ever-increasing computational capability of the digital processors, and suffer \revOmid{from the} issues regarding management and distribution of secret keys~\cite{6739367, GKJPO13}. Alternatively, physical layer security takes advantage of physical characteristics of the communication medium in order to provide a secure data exchange between the information transmitter (Alice) and the legitimate receiver (Bob). In the seminal work by Wyner \cite{6772207}, the concept of secrecy capacity is introduced for a three-node degraded wiretap channel, as the maximum information rate that can be exchanged \revPeterNew{under the condition of perfect secrecy}. It is shown that a positive secrecy capacity is achievable when the physical channel to the eavesdropper is \revOmid{weaker than} the channel to the legitimate receiver. The arguments of \cite{6772207} {have since been extended} in the directions of secrecy rate region analysis for various wiretap channel models \cite{4529282,5730586,5961840}, construction of capacity achieving channel codes \cite{6283924,6034749, 5545658}, as well as signal processing techniques for enhancing the \revOmid{secrecy capacity}, see \cite{4960114, 4543070, 7775096, 6872551, 6266769, 5940246} and references therein. \revOmid{As a result, physical layer security is considered as a \revPeterNew{possible solution} for providing information security in different use cases \revPeterNew{for} future wireless communication systems, e.g.,~\cite{phy_5G1, phy_5G2}.}\par

As a promising method to enhance the secrecy capacity of wireless \revOmid{communication} systems, Goel and Negi \cite{4543070} have introduced the idea of friendly jamming, i.e., the \revOmid{transmission of an artificially-generated noise} (AN) with the intention of degrading the decoding capability \revPeter{of} the eavesdropper. This can be implemented via \emph{i)} the joint transmission of information and AN from Alice, however, requiring an effective beamforming capability at Alice and sharing the communication and jamming resources~\cite{4960114, 4543070}, \emph{ii)} utilization of external (cooperative) jammers, nonetheless, resulting in the issues of jammer mobility, synchronization, and trustworthiness \cite{5730586, 6266769, 5940246,  7775096, 6872551}, and \emph{iii)} via the application of full-duplex (FD) nodes~\cite{GKJPO13}. \revOmid{An FD transceiver is capable of transmission and reception at the same time and frequency \revPeterNew{band}, however, suffering from the \revOmid{strong SI} \revPeterNew{which is caused by its own transmitter}}. \revOmid{The recently developed methods for self-interference cancellation (SIC) \cite{FD_SIC_Passive,FD_SIC_khandani,BMK:13}, have demonstrated \revPeterNew{practical implementations} of FD transceivers in the last few years, employing various techniques including passive isolation, antenna design and placement, and analog and digital signal processing methods. The implementation in \cite{BMK:13} \revPeterNew{achieves a SIC level of $110$,} complying with the IEEE $802.11$ac (WiFi) standard. The extension of the aforementioned works for multiple antenna transceivers have been studied in \cite{FD_SIC_softnull, Bharadia:14, FD_SIC_princeton} for different setups, and motivated a wide range of related applications, see \cite{HBCJMKL:14} and the references therein.} In this regard, \revPeterNew{a} FD Bob can act as a friendly jammer, while simultaneously receiving information from Alice. Note that \revPeterNew{a} FD jammer does not occupy the communication resources from Alice, nor does it rely on the external helpers, resolving the related drawbacks. The problems regarding secrecy rate region analysis and resource optimization \revPeter{have hence been} addressed in \cite{GKJPO13, 7945480, 7339654, ZGYZ:14, 6787008, 7792199, secrecy_ICC_taghizadeh}. It is observed that a significant gain is achievable, in terms of the resulting secrecy capacity via the utilization of \revPeterNew{\revOmid{a}} FD jamming strategy under the condition that the \revOmid{SI} signal can be \revPeter{attenuated effectively}.\par             
\revOmid{Although the available literature introduces a gainful utilization of FD transceivers for enhancing the \revPeterNew{system's} secrecy capacity, the aforementioned gain comes at the expense of a higher power consumption due to \emph{i)} the implementation of \revOmid{a} SIC scheme at the FD transceiver, incurring additional digital processing and analog circuitry, \emph{ii)} the power consumed for the transmission of AN, as well as \emph{iii)} the degrading impact of residual self-interference on the desired communication link.} As a result, it is not clear how the FD jamming-enabled systems perform in terms of secrecy energy efficiency (SEE). Note that the issue of energy efficiency is recently raised as a key criteria in the design of wireless communication systems. This is justified \revPeterNew{by the exponential increase of information and communication technology (ICT) services}, currently generating around $5\%$ of global $\text{CO}_2$ emissions and \revPeterNew{with an expected share of $75\%$ for wireless systems in $2020$}, which calls for novel energy efficient ICT solutions \cite{7446253, 5978416}. Hence, it is the main purpose of this paper to investigate if and how the application of FD jammers can enhance the system's SEE, in terms of the securely communicated bits per Joule (SBPJ).\vspace{-3mm} \vspace{-1mm}
\vspace{-1mm}
\subsection{Contribution and paper organization}
\begin{itemize}[leftmargin=*]
\item In this work, \revOmid{we study \revPeterNew{an}} SEE [SBPJ] maximization problem for a general MIMOME setup, where the legitimate receiver is capable of FD jamming. \revPeterNew{This stays in contrast to} available designs which utilize FD transceivers for improving the secrecy capacity~\cite{7339654, ZGYZ:14, 6787008, 7792199, secrecy_ICC_taghizadeh, 7945480} or the studies on the SEE of half-duplex (HD) networks~\cite{zappone2016energy,7094604, 6199997, 6476945}. Due to the intractable problem structure, \revOmid{an iterative algorithm} is proposed, with a guaranteed convergence to a point satisfying \revPeterNew{the Karush-Kuhn-Tucker (KKT) conditions of the original problem}. 

\item \revPeterNew{The joint utilization of FD capabilities, both on Alice and Bob, for jamming and bi-directional information exchange, shows additional potential for SEE improvement}. This is grounded on the fact that, firstly, the FD jamming power is reused for both communication directions, resulting in power-efficient jamming, and secondly, the coexistence of two communication directions on the same channel may degrade Eve's decoding capability. Motivated by this, the \revOmid{proposed iterative algorithm is extended for} \revPeterNew{a bidirectional FD} setup. 

\item In order to account for \revOmid{channel state information (CSI)} uncertainties, the consideration of statistical CSI regarding the channels to Eve has been introduced in \cite{lin2013secrecy, zappone2016energy}, considering HD nodes. However, the aforementioned works limit the studied setups to a single antenna Eve, where CSI statistics follow a specific fast-fading nature. In this work \revPeterNew{\revOmid{an}} SEE maximization problem is studied for \revPeterNew{\revOmid{a}} FD-enabled MIMOME setup, where the channels to Eve follow an arbitrary statistical distribution. Note that unlike the fast-fading \revPeterNew{condition}, which assumes the CSI is not available due to mobility, the consideration of an arbitrary statistical distribution also accounts for the scenarios \revPeterNew{where Eve is stationary, but Eve's CSI cannot be obtained due to the lack of collaboration from Eve}. Hence, we propose a successive selection and statistical lower bound maximization (SSSLM) algorithm, utilizing a combination of \revOmid{sample average} approximation~\cite{kim2015guide}, \revPeterNew{\revOmid{and}} successive lower bound approximation method \cite{razaviyayn2013unified}, with the goal of maximizing the \revPeterNew{statistical average SEE}. The algorithm is proven to converge to a point satisfying the KKT conditions \revPeterNew{of the original problem}. 
\end{itemize}
The numerical results indicate only a marginal SEE gain, \revPeterNew{through} the utilization of FD jamming, for a wide range of system \revPeterNew{parameters}. However, the observed SEE gain is notable for \revPeterNew{systems with a small distance between the FD node and the eavesdropper, a high \revPeter{signal-to-noise ratio (SNR)}, or for a bidirectional FD communication setup, if SI can efficiently be mitigated.}

\revOmid{The studied system model is defined in Section~\ref{sec:model}. When \revPeterNew{only the legitimate receiver (Bob)} is capable of FD jamming, an SEE maximization framework is introduced in Section~\ref{sec_SEE_max}, and then extended \revPeterNew{to a bidirectional} FD communication setup in Section~\ref{sec_SEE_max_BD}. \revPeterNew{In} Section~\ref{sec_SSUM}, \revPeterNew{an} SEE maximization framework is introduced \revPeterNew{for the case} when the channels to the eavesdropper are not accurately known. The behavior of the proposed algorithms, as well as the impact of the FD jamming on the SEE performance are numerically studied in Section~\ref{sec:simulations}. This paper is concluded in Section~\ref{sec:conclusion} by summarizing the main results.} 
\subsection{Mathematical notation}
Throughout this paper, column vectors and matrices are denoted as lower-case and {upper-case} bold letters, respectively. Mathematical expectation, trace, determinant, {and} Hermitian transpose are denoted by $ \mathbb{E}(\cdot), \; {\text{ tr}}(\cdot), \;  |\cdot|,$ {and} $(\cdot)^{ H},$ respectively. The Kronecker product is denoted by $\otimes$. The identity matrix with dimension $K$ is denoted as ${\ma I}_K$ and ${\rm vec}(\cdot)$ operator stacks the elements of a matrix into a vector. {Moreover,} $(\cdot)^{-1}$ represents the inverse of a matrix and $||\cdot||_{2},||\cdot||_{\text{F}}$ {respectively represent the Euclidean and Frobenius norms}. {\MinorRR{$\ma{0}_{m \times n}$ represents an all-zero matrix with size $m \times n$.}} $\text{diag}(\cdot)$ returns a diagonal matrix by putting the off-diagonal elements to zero. \revPeter{$\bot$ denotes statistical independence.} The set $\mathbb{F}_K$ is defined as $\{1,2,\ldots,K\}$, and $|\mathbb{X}|$ denotes the size of the set $\mathbb{X}$. \revPeter{The set of positive real numbers,} the set of complex numbers, and the set of all positive semi-definite matrices with Hermitian symmetry are denoted by $\mathbb{R}^+$, $\compl$ and $\mathcal{H}$, respectively. $a^\star$ indicates the value of $a$ for which optimality holds. The value of $\{x\}^+$ is equal to $x$, if positive, and zero otherwise. \revPeter{Furthermore, $\mathcal{CN} \left( \ma{x}, \ma{X} \right)$ denotes the complex normal distribution with mean $\ma{x}$ and covariance $\ma{X}$.}

%% file: Sections/sec_systemmodel.tex
\begin{figure}[!t] 
    \begin{center}
        \includegraphics[angle=0,height=4cm, width=0.5\columnwidth]{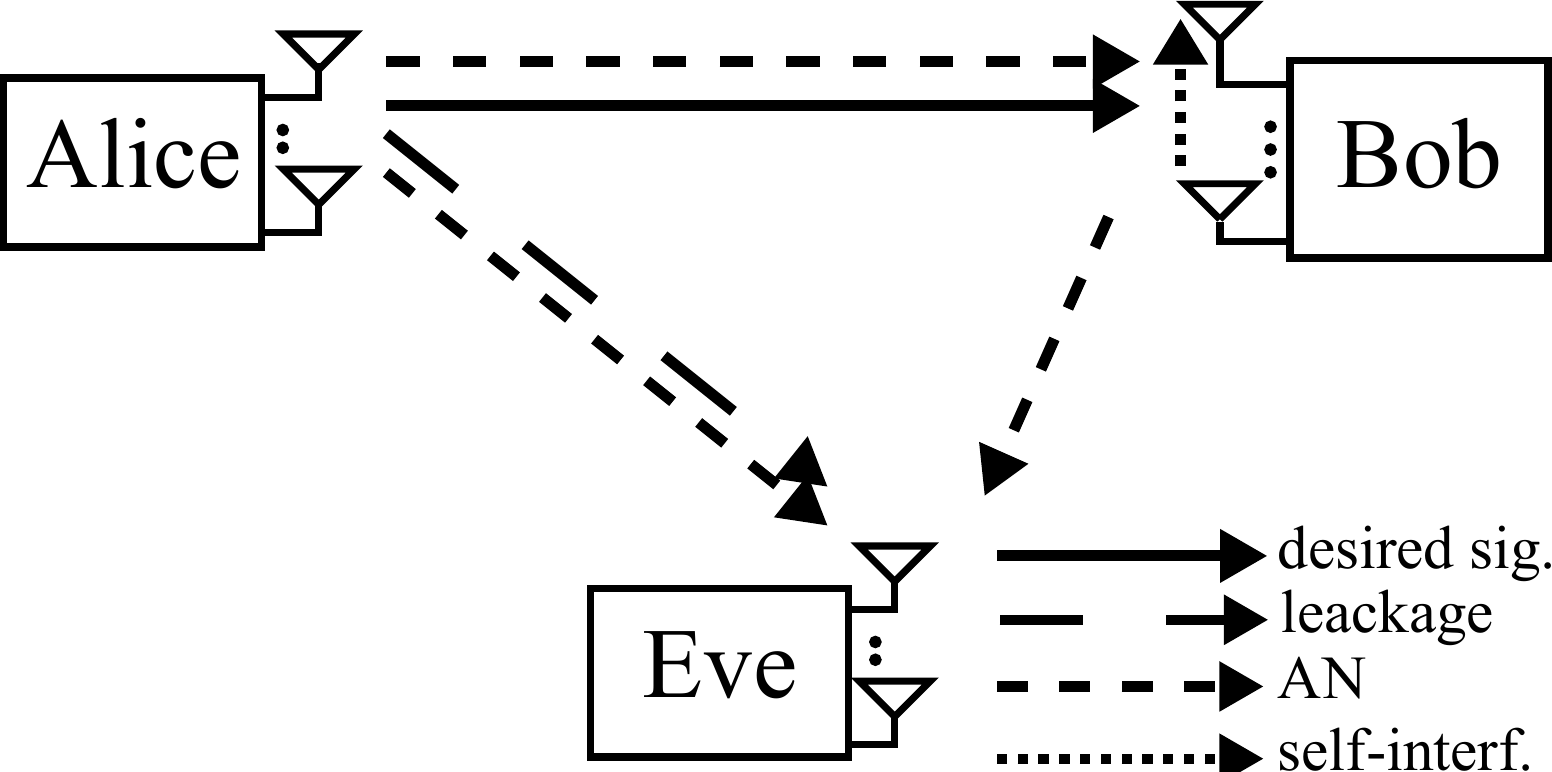} \vspace{-5mm} 
    \caption{The studied wiretap channel. Alice and the FD-Bob are jointly enabled with jamming capability.} \label{fig_model}
    \end{center}  
\end{figure} 
We consider a \revPeter{MIMOME} wiretap channel \revPeter{that consists of} a legitimate transmitter, i.e., Alice, a legitimate receiver, i.e., Bob, and an eavesdropper, i.e., Eve, see Fig.~\ref{fig_model}. Alice and Eve are equipped with $N_{A}$ transmit and $M_{E}$ receive antennas, respectively. Bob is respectively equipped with $N_{B}$ and $M_{B}$ transmit and receive antennas, and is capable of FD operation. \MinorRevOmid{Note that the FD operation at Bob does not indicate an FD communication, as Alice is assumed to be an HD node in this section. However, the FD capability at Bob is used to send a jamming signal while receiving information, and thereby to degrade the decoding capability at the undesired receiver. The joint FD operation at Alice and Bob, facilitating an FD communication, is later introduced in Section~\ref{sec_SEE_max_BD}.} Channels are assumed to follow a quasi-stationary\footnote{It means that the channel remains constant within a frame, but may change from one frame to another.} and flat-fading model. In this regard, channel from Alice to Bob, Alice to Eve, and Bob to Eve (jamming channel) are respectively denoted as $\ma{H}_{ab}\in \compl^{M_{B} \times  N_{A}}$, $\ma{H}_{ae} \in \compl^{ M_{E} \times N_{A} }$, $\ma{H}_{be} \in \compl^{M_{E}\times N_{B}}$. The channel from Bob to Bob, i.e., self-interference channel, is denoted as $\ma{H}_{bb}\in \compl^{M_{B} \times N_{B}}$. 

\subsection{Signal model} \label{sec_model_signalmodel}
The transmission from Alice includes the information-containing signal, intended for Bob, and \revPeter{an} AN signal\footnote{\revOmid{Unlike the data symbols, which follow a known constellation, the AN is generated from a pseudo-random sequence which is not known to the receivers, see \cite[Section~III]{4543070}.  This ensures that Eve cannot decode the AN.}}, intended to degrade the reception by Eve. This is expressed as 
\begin{align} \label{eq_model_tx_alice}
\ma{x}_a =  \underbrace{\ma{q}_a + \ma{w}_a}_{\ma{u}_a} + \ma{e}_{\text{tx},a}, 
\end{align}  
where $\ma{u}_a \in \compl^{{N_{A}}}$ is the intended transmit signal, $\ma{q}_a \sim \mathcal{CN} \left( \ma{0}_{N_{A}}, \ma{Q}_a \right)$ and $\ma{w}_a \sim \mathcal{CN} \left( \ma{0}_{N_{A}}, \ma{W}_a \right)$ respectively represent the information-containing and AN signal, and $\ma{x}_a \in \compl^{N_{A}}$ is the combined transmitted signal from Alice. The transmit distortion, denoted as $\ma{e}_{\text{tx},a} \in \compl^{N_{A}}$ models collective \revPeter{impact} of transmit chain inaccuracies, e.g., \revPeterNew{\revOmid{digital-to-analog \revPeter{converter} noise}}, power amplifier (PA) noise, oscillator phase noise, see Subsection~\ref{sec_model_diststatistics} for more details. Note that the role of hardware inaccuracies becomes important in a system with FD transceivers, due to the impact of \revPeter{a} strong self-interference channel. Similar to the transmission from Alice, the transmission of AN by Bob is expressed as 
\begin{align} \label{eq_model_tx_bob}
\ma{x}_b =  \ma{w}_b + \ma{e}_{\text{tx},b}, 
\end{align}    
where $\ma{w}_b \sim \mathcal{CN} \left( \ma{0}_{N_{B}}, \ma{W}_b \right)$ is the transmitted artificial noise and $\ma{e}_{\text{tx},b} \in \compl^{N_{B}}$ represents the transmit distortions from Bob. Via the application of (\ref{eq_model_tx_alice}) and (\ref{eq_model_tx_bob}) the received signal at Eve is expressed as 
\begin{align} \label{eq_model_rx_eve}
\ma{y}_e &=  \ma{H}_{ae}\ma{x}_a +  \ma{H}_{be}\ma{x}_b + \ma{n}_e = \ma{H}_{ae}\ma{q}_a + \ma{c}_e ,
\end{align}
where $\ma{n}_e \sim \mathcal{CN} \left( \ma{0}_{M_E}, \sigma_{\text{n},e}^2 \ma{I}_{M_E} \right)$ is the additive thermal noise and 
\begin{align} \label{eq_model_collective_interf}
\ma{c}_e := \ma{H}_{ae}\ma{w}_a + \ma{H}_{be}\ma{w}_b + \ma{H}_{ae}\ma{e}_{\text{tx},a} + \ma{H}_{be}\ma{e}_{\text{tx},b} + \ma{n}_e
\end{align}
is the collective interference-plus-noise at Eve. Similarly, the received signal at Bob is formulated as 
\begin{align} \label{eq_model_rx_bob}
\ma{y}_b =  \underbrace{\ma{H}_{ab}\ma{x}_a +  \ma{H}_{bb}\ma{x}_b  + \ma{n}_b}_{=: \ma{u}_b} + \ma{e}_{\text{rx},b},  
\end{align} 
where $\ma{n}_b \sim \mathcal{CN} \left( \ma{0}_{M_{B}}, \sigma_{\text{n},b}^2\ma{I}_{M_{B}} \right)$ is the additive thermal noise, and $\ma{u}_b$ is the received signal, assuming perfect hardware operation. Similar to the transmit side, the receiver side distortion, denoted as $\ma{e}_{\text{rx},b} \in \compl^{M_{B}}$, models the collective impact of receiver chain inaccuracies, e.g., \revPeterNew{\revOmid{analog-to-digital \revPeter{converter} noise}}, oscillator phase noise, and automatic gain control error, see Subsection~\ref{sec_model_diststatistics}. Note that $\ma{y}_b$ includes the received self-interference signal at Bob, originating from the same transceiver. Hence, the \emph{known}, i.e., distortion-free, part of the self-interference can be subtracted applying \revPeter{a} SIC method~\cite{Bharadia:14,BMK:13}. The received signal at Bob, after the application of SIC is hence written as 
\begin{align} \label{eq_model_rx_bob_afterSIC}
\tilde{\ma{y}}_b & =  {\ma{y}}_b - \ma{H}_{bb}\ma{w}_b  \nonumber \\ &= \ma{H}_{ab} \ma{x}_a  +  \ma{H}_{bb} \ma{e}_{\text{tx},b} + \ma{e}_{\text{rx},b} + \ma{n}_b = \ma{H}_{ab}\ma{q}_a + \ma{c}_b ,	
\end{align} 
where 
\begin{align} \label{eq_model_collective_interf}
\ma{c}_b := \ma{H}_{ab}\ma{w}_a + \ma{H}_{ab}\ma{e}_{\text{tx},a} + \ma{H}_{bb}\ma{e}_{\text{tx},b} + \ma{e}_{\text{rx},b} +  \ma{n}_b,
\end{align}          
is the collective interference-plus-noise at Bob.                         
\vspace{-2mm}\subsection{Distortion signal statistics}\label{sec_model_diststatistics}
Similar to \cite{DMBS:12}, we model the impact of transmit (receive) chain inaccuracies by injecting Gaussian-distributed and independent distortion terms at each antenna\footnote{Eve is assumed to operate with a zero-distortion hardware, considering a worst-case scenario.}. Moreover, the variance of the distortion signals are proportional to the power of the intended transmit (receive) signal at the corresponding chain. This model is elaborated in \cite[Subsections~C]{DMBS:12}, \cite{MITTX:98, MITTX:08} regarding the characterization of hardware impairments in transmit chains, and in \cite[Subsections~D]{DMBS:12}, \cite{MITRX:05} for the receiver chains. This is expressed in our system as
\begin{align} 
\ma{e}_{\text{tx},a}   \sim  \mathcal{CN} & \Big( \ma{0}_{N_{A}},  \kappa_a \text{diag} \Big( \mathbb{E} \left\{\ma{u}_a \ma{u}_a^H \right\} \Big) \Big),   \;\; \ma{e}_{\text{tx},a}  \bot  \ma{u}_a , \label{eq_model_e_tx_a}   \\
\ma{e}_{\text{tx},b}   \sim  \mathcal{CN} & \Big( \ma{0}_{N_{B}},  \kappa_b \text{diag} \Big(\mathbb{E} \left\{\ma{w}_b \ma{w}_b^H \right\} \Big) \Big),  \;\; \ma{e}_{\text{tx},b}  \bot  \ma{w}_b, \label{eq_model_e_tx_b}   \\ 
\ma{e}_{\text{rx},b}   \sim  \mathcal{CN} & \Big( \ma{0}_{M_{B}},  \beta_b \text{diag} \Big( \mathbb{E} \left\{\ma{u}_b \ma{u}_b^H \right\} \Big) \Big), \;\; \ma{e}_{\text{rx},b}  \bot  \ma{u}_{b},    \label{eq_model_e_rx_b} 
\end{align} 
where $\kappa_a, \kappa_b, \beta_b \in \mathbb{R}^+$ are distortion coefficients, relating the variance of the distortion terms to the intended signal power, and $\ma{u}_a$ and $\ma{u}_b$ are defined in (\ref{eq_model_tx_alice}) and (\ref{eq_model_rx_bob}), respectively. For further elaborations on the used distortion model please see \cite{DMBS:12, DMBSR:12, ALRWW:14, XaZXMaXu:15}, and the references therein.  
%
\subsection{Power consumption model}\label{sec_model_powerconsumption}
The consumed power of a wireless transceiver can be segmented into three parts. First, the power consumed at the PA, which is related to the effective transmit power via \revPeterNew{the} PA efficiency, see \cite[Eq.~(2)]{SGB:04}. Secondly, the zero-state power, i.e., the power consumed by other circuit blocks, independent from \revPeterNew{the} transmission status\footnote{This includes, e.g., the power consumed at receiver chain, and for base band processing.}, see \cite[Eq.~(3)]{SGB:04}. And finally, the power consumed for the implementation of \revPeter{a} SIC scheme, \revPeterNew{to enable} FD operation. The aforementioned power varies for different SIC methods, and by definition, is not relevant for HD transceivers. The consumed power for Alice and Bob \revPeter{can hence be} expressed as
\begin{figure*}[!t]
\normalsize
\setcounter{mytempeqncnt1}{\value{equation}}
\setcounter{equation}{14}
{\small{ \begin{align}
\ma{\Sigma}_b =  \mathbb{E}\{\ma{c}_b\ma{c}_b^H\} & = \ma{H}_{ab}\ma{W}_a\ma{H}_{ab}^H + \kappa_a \ma{H}_{ab} \text{diag} \left(\ma{Q}_a+\ma{W}_a\right) \ma{H}_{ab}^H + \kappa_b \ma{H}_{bb} \text{diag} \left( \ma{W}_b \right) \ma{H}_{bb}^H \nonumber\\ 
& \;\;\;\;\;\;\;\;\;\;\;\;\;\;\;\;\;\;\;\;\; + \beta_b \text{diag} \Big( \ma{H}_{ab} \left( \ma{Q}_a + \ma{W}_a \right) \ma{H}_{ab}^H + \ma{H}_{bb} \ma{W}_b \ma{H}_{bb}^H + \sigma_{\text{n},b}^2 \ma{I}_{M_{B}} \Big) + \sigma_{\text{n},b}^2 \ma{I}_{M_{B}},   \label{eq_model_Sigma_b}\\
\ma{\Sigma}_e =  \mathbb{E}\{\ma{c}_e\ma{c}_e^H\} & =  \ma{H}_{ae}\ma{W}_a\ma{H}_{ae}^H + \ma{H}_{be}\ma{W}_b \ma{H}_{be}^H + \kappa_a \ma{H}_{ae} \text{diag} \left(\ma{Q}_a+\ma{W}_a\right) \ma{H}_{ae}^H + \kappa_b \ma{H}_{be} \text{diag} \left( \ma{W}_b \right) \ma{H}_{be}^H + \sigma_{\text{n},e}^2 \ma{I}_{M_{E}}.  \label{eq_model_Sigma_e}
\end{align} }}  
\setcounter{equation}{\value{mytempeqncnt1}}
\hrulefill
\end{figure*}
\setcounter{equation}{10}
\revOmid{\begin{align} \label{eq_model_power_Alice}
P_{A} \left(\ma{Q}_a, \ma{W}_a\right) =   \frac{1 + \kappa_a}{\mu_A} \text{tr}\left(\ma{Q}_a + \ma{W}_a\right) + P_{A,0}, \;\; P_{A} \leq P_{A,\text{max}}
\end{align} 
and
\begin{align} \label{eq_model_power_Bob}
P_{B} \left( \ma{W}_b \right) =   \frac{1 + \kappa_b}{\mu_B} \text{tr}\left(\ma{W}_b\right) + P_{B,0} + P_{\text{FD}} , \;\; P_{B} \leq P_{B,\text{max}}.
\end{align}} 
In the above arguments, $P_{\mathcal{X}}, P_{\mathcal{X},0}, \mu_{\mathcal{X}}$, and $P_{\mathcal{X},\text{max}}$, where $\mathcal{X} \in \{A,B\}$, respectively represent the consumed power, the zero-state power, PA efficiency, and the maximum allowed power consumption for each node. The additional required power for the implementation of \revPeterNew{a} SIC scheme is denoted by $P_{\text{FD}}$. From (\ref{eq_model_power_Alice}) \revPeterNew{and} (\ref{eq_model_power_Bob}), the total system power consumption is obtained as \vspace{-3mm}
\revOmid{ \begin{align} \label{eq_model_power_total}
P_{\text{tot}} \left(\ma{Q}_a, \ma{W}_a, \ma{W}_b\right) = P_{A}\left(\ma{Q}_a, \ma{W}_a\right) +  P_{B}\left( \ma{W}_b \right).
\end{align} } 
\subsection{Secrecy energy efficiency}\label{sec_model_secrecy_EE}
\revOmid{Following \cite{4543070, 5961840, GKJPO13}, the achievable secrecy rate\footnote{\revOmid{The system\revPeterNew{'s} secrecy capacity is lower bounded by all achievable secrecy rates, resulting from different choices of transmit covariance matrices, see \cite[Theorem~1]{5961840}, \cite[Equation~(6)]{4543070}.}} for Alice-Bob communication} is expressed as $C_{{ab}} =  \left\{ \tilde{C}_{{ab}} \right\}^{+}, 
$ such that
\begin{align}   \label{eq_model_ab_cap}
\tilde{C}_{{ab}}=  \text{log} \left|  \ma{I} +  \ma{H}_{ab}\ma{Q}_a\ma{H}_{ab}^H \ma{\Sigma}_b^{-1} \right| - \text{log} \left|  \ma{I} +  \ma{H}_{ae}\ma{Q}_a\ma{H}_{ae}^H \ma{\Sigma}_e^{-1} \right|, 
\end{align} \setcounter{equation}{16}
where $\ma{\Sigma}_b$, $\ma{\Sigma}_e$ are given in (\ref{eq_model_Sigma_b}), (\ref{eq_model_Sigma_e}), and represent the covariance of the interference-plus-noise terms at Bob and Eve, respectively. \revOmid{The SEE}, as a measure of securely communicated information per energy unit, is consequently expressed as
\begin{align} \label{eq_model_Sec_cap}
\text{SEE} = \frac{C_{{ab}}}{P_{\text{tot}}}. 
\end{align}
\revPeter{It} is the intention of the remaining \revPeterNew{sections} of this paper to improve the efficiency of the defined wiretap channel, in terms of the SEE, and provide comparison \revPeterNew{between FD and} usual HD strategies.  
\MinorRevOmid{\begin{remark}
In this part, we have introduced an MIMOME wiretap channel, where Bob is capable of FD operation and sends a jamming signal in order to improve the information security. However, this does not facilitate an FD communication, as Alice remains an HD node. This setup is relevant in the practical asymmetric scenarios, e.g., the uplink of an FD cellular communication system~\cite{7463025}, where users are usually not capable of FD operation. The setup with the joint FD operation at Alice and Bob, facilitating a joint jamming and an FD bidirectional communication, is later discussed in Section~\ref{sec_SEE_max_BD}.
\end{remark}
\begin{remark}
\revOmid{In this part we assume the availability of exact CSI for all channels, relevant to the scenarios with a collaborative eavesdropper, e.g.,~\cite{UntRel:1, UntRel:2}. The scenario with the availability of partial CSI is discussed in Section~\ref{sec_SSUM}. }
\end{remark} }
\begin{remark}
\revOmid{Unlike the data symbols, which follow a known constellation, the AN is generated from a pseudo-random sequence which is not known to the receivers, see \cite[Section~III]{4543070}. This ensures that Eve may not decode the AN \revPeterNew{and hence, cannot cancel the interference caused by the AN transmissions}}
\end{remark}

%% file: Sections/sec_SEE_max.tex
In this part we intend to enhance the system SEE, assuming the availability of CSI for all channels. The corresponding optimization problem is defined as 
{\begin{subequations}  \label{problem_SEE_max_1}
{ \begin{align}
\underset{ \ma{Q}_a, \ma{W}_a, \ma{W}_b }{\text{max}} \;\; & \text{SEE}\left( \ma{Q}_a, \ma{W}_a, \ma{W}_b \right)  \label{problem_SEE_max_1_a} \\             
\text{s.t.} \;\;\;\;\;\; & {P_{A} \left(\ma{Q}_a, \ma{W}_a\right) \leq P_{A,\text{max}}, \;\;   P_{B} \left( \ma{W}_b \right) \leq P_{B,\text{max}}, \;\; \ma{Q}_a, \ma{W}_a, \ma{W}_b \in \mathcal{H}, }\label{problem_SEE_max_1_c}
\end{align}}      
\end{subequations}} \hspace{-3mm}
where (\ref{problem_SEE_max_1_c}) represent the power constraints at Alice and Bob, see~(\ref{eq_model_power_Alice}),~(\ref{eq_model_power_Bob}). The defined problem in (\ref{problem_SEE_max_1}) is not tractable in the current form, due to the non-convex and non-smooth objective. In order to obtain a tractable structure, without loss of optimality, we remove the non-linear operator $\{\}^+$ from the definition of SEE\footnote{Note that at the optimality of (\ref{problem_SEE_max_1}), the resulting $C_{\text{s}}$, and consequently the SEE is non-negative. This is since a non-negative SEE is immediately obtained by setting $\ma{Q}_a = \ma{0}$, \MinorRevOmid{see Lemma~\ref{lemma_BD_Positive_C} for a generalized proof.}}. The modified SEE, named $\text{SEE}_p$ hereinafter, can \revPeter{hence be} formulated as 
\begin{align}  \label{eq_SEE_max_SEE_p}                                 
&\text{SEE}_p\left( \ma{Q}_a, \ma{W}_a, \ma{W}_b \right) = \frac{ \Sigma_{\mathcal{X} \in \{ b,e\}} \alpha_\mathcal{X}\left( \text{log} \left|  \ma{\Sigma}_\mathcal{X} +  \ma{H}_{a\mathcal{X}}\ma{Q}_a\ma{H}_{a\mathcal{X}}^H  \right| - \text{log} \left|  \ma{\Sigma}_\mathcal{X}  \right|  \right)}{ P_{\text{tot}} \left( \ma{Q}_a, \ma{W}_a, \ma{W}_b \right) },
\end{align}        
where $\alpha_b = 1$ and $\alpha_e = -1$. It is observed that $\text{SEE}_p$ is a difference of concave (DC) over affine fractional function which is intractable in the current form. In the following we propose a successive general inner approximation algorithm (SUIAP) to obtain a {\MinorRR{KKT solution to (\ref{problem_SEE_max_1}), i.e., a solution satisfying the \revOmid{KKT optimality conditions}}}. 
\revOmid{\subsection{SUIAP} \label{subsec_SUIAP}
The proposed SUIAP algorithm consists of two nested loops. The detailed procedure is explained in the following. 
\subsubsection{Initialization} \label{sec_SUIAP_init}
{In this section we briefly discuss the initialization of Algorithm~\ref{SUIAP}.} We separate the choice of spatial beams and power allocation for different transmissions, in order to obtain a fast solution.  
\subsubsection{Spatial adjustment}
The role of the transmit spatial adjustment is to direct the transmit signal to the desired receiver, while preventing leakage to the undesired directions. This is written as the following maximization
\begin{align}   \label{eq_op_SEE_max_ratio_initialization}                                 
\underset{\ma{Q} }{ \text{max}} \;\;  \frac{\text{tr} \left( \ma{F} \ma{Q} \ma{F}^H  \right) + \nu_f}{\text{tr} \left( \ma{G} \ma{Q} \ma{G}^H \right) +  \nu_g}, \;\; \text{s.t.} \;\; \text{tr} \left(\ma{Q}\right) = 1,
\end{align}  
where $\ma{Q}$ represents the normalized covariance matrix, $\ma{F}$ and $\ma{G}$ are the desired and undesired channels, and $\nu_f, \nu_g$ are the noise variances at the desired and undesired receivers, respectively. An optimal solution to (\ref{eq_op_SEE_max_ratio_initialization}) can be obtained as 
\begin{align}  \label{eq_SEE_max_ratio_initialization}                                 
\text{vec}\left( {\ma{Q}^\star}^{\frac{1}{2}} \right) = \mathcal{P}_{\text{max}} \left( \left( \ma{I}\otimes \ma{G}^H\ma{G}    + \nu_g \ma{I} \right)^{-1}  \left( \ma{I}\otimes \ma{F}^H\ma{F} + \nu_f \ma{I}  \right) \right).
\end{align}
where $\mathcal{P}_{\text{max}} \left( \cdot \right)$ calculates the dominant eigenvector. Note that the above approach is applied separately for the spatial adjustment of $\ma{Q}_a,\ma{W}_a$ and $\ma{W}_b$. The corresponding desired and undesired channels are defined in Appendix~\ref{appendix_init_spatialadjustent}.
  
\subsubsection{Power allocation}
The transmit power adjustment for $\ma{Q}_a,\ma{W}_a$ and $\ma{W}_b$ is obtained by applying the normalized covariance in the previous part as the basis. Afterwards, the power for each transmission is optimized to maximize $\text{SEE}_p$, see Appendix~\ref{appendix_init_powAdj}.
 
\subsubsection{Outer loop}                       
In each outer iteration, the optimization problem (\ref{problem_SEE_max_1}) is approximated by replacing the objective with an effective lower bound to $\text{SEE}_p$, following the successive inner approximation (SIA) framework \cite{marks1978technical}. This is implemented by applying the inequality 
\begin{align}  \label{eq_SEE_max_SEE_Taylor}                                 
- \text{log} \left|  \ma{X}\right| \geq - \text{log} \left|  \ma{Y} \right| + \text{tr}\left( \ma{Y}^{-1} \left(  \ma{Y} - \ma{X} \right) \right)
\end{align}  
obtained from the first-order Taylor approximation of the convex terms $- \text{log} \left| \ma{X} \right|$ at the point $\ma{X} = \ma{Y}$. The approximated optimization problem at the $l$-th outer iteration is consequently expressed as 
\begin{align}  \label{eq_SEE_lower_bound_SIA}   
\underset{ \mathbb{Q}^{[l]} }{\text{max}} \;\; & \underbrace{C_{LB,ab} \left(\mathbb{Q}^{[l]}, {\mathbb{Q}^{[l-1]}}^\star \right)/P_{\text{tot}} \left(\mathbb{Q}^{[l]}\right) }_{\leq  \text{SEE}_p\left( \mathbb{Q}^{[l]} \right)}  \;\;  \text{s.t.} \;\; \text{(\ref{problem_SEE_max_1_c})} ,                  
\end{align} 
\begin{figure*}[!t]
\normalsize
{{\revOmid{{{\begin{align}
&  C_{LB,ab}  \left( \mathbb{Q}^{[i]} , {\mathbb{Q}}^{[j]} \right) :=   \text{log} \left|  \ma{\Sigma}_b \left( {\mathbb{Q}}^{[i]} \right) +  \ma{H}_{ab}\ma{Q}_a^{[i]}\ma{H}_{ab}^H  \right|  - \text{log} \left|  \ma{\Sigma}_b\left( {\mathbb{Q}}^{[j]} \right) \right|  + \text{log} \left| \ma{\Sigma}_e\left( {\mathbb{Q}}^{[i]} \right) \right| \nonumber \\
&   +\hspace{-0mm} \text{tr}\Big(\hspace{-0mm} \Big( \ma{\Sigma}_e\left( {\mathbb{Q}}^{[j]} \right) +  \ma{H}_{ae}\ma{Q}_a^{[j]}\ma{H}_{ae}^H \Big)^{-1} \Big( \hspace{-0mm} \ma{\Sigma}_e\left( {\mathbb{Q}}^{[j]} \right) -   \ma{\Sigma}_e\left( {\mathbb{Q}}^{[i]} \right) + \ma{H}_{ae}\Big( \hspace{-0mm} \ma{Q}_a^{[j]} - \ma{Q}_a^{[i]} \hspace{-0mm} \Big) \ma{H}_{ae}^H\Big)\hspace{-0mm} \Big) \nonumber \\
&  + \text{tr}\left( \left( \ma{\Sigma}_b\left( {\mathbb{Q}}^{[j]} \right) \right)^{-1}  \left(  \ma{\Sigma}_b\left( {\mathbb{Q}}^{[j]} \right) -  \ma{\Sigma}_b\left( {\mathbb{Q}}^{[i]} \right) \right) \right)   - \text{log} \left|  \ma{\Sigma}_e\left(\hspace{-0mm} {\mathbb{Q}}^{[j]}\hspace{-0mm} \right) \hspace{-0mm}+ \hspace{-0mm} \ma{H}_{ae}\ma{Q}_a^{[j]}\ma{H}_{ae}^H  \right|  \hspace{-0mm}   \label{eq_SEE_max_SEE_Approx_Dincklebach} 
\end{align} }}} }}
\hrulefill
\vspace*{-5mm}
\end{figure*}
where $\mathbb{Q}^{[X]}:= \left\{ \ma{Q}_a^{[X]}, \ma{W}_a^{[X]}, \ma{W}_b^{[X]} \right\}$, with $X$ specifying an iteration instance. Moreover, $C_{LB,ab}$ is given in (\ref{eq_SEE_max_SEE_Approx_Dincklebach}) and ${\mathbb{Q}^{[l-1]}}^\star$ represents the obtained solution at the previous outer iteration. 
 
\subsubsection{Inner loop}
The inner loop is dedicated to optimally solve the approximated problem at each outer iteration via the well-known Dinkelbach's algorithm \cite{dinkelbach1967nonlinear}, as (\ref{eq_SEE_lower_bound_SIA}) belongs to the class of concave-over-affine fractional programs \cite{beckenbach1937generalized}. In particular, the optimum solution is obtained via a sequence of parametric variable updates, see Appendix~\ref{appendix_dinkelbach} for elaboration on the related class of fractional programs and the detailed procedure. The main steps associated with the Dinkelbach's algorithm, i.e., Steps~$3$ and $5$ from Algorithm~\ref{alg_Dinkelbach}, can be expressed as the following updates:
\begin{align}                                 
{ \mathbb{Q}^{[l,k]} }^\star  & \leftarrow  \underset{ \mathbb{Q}^{[l,k]} }{\text{argmax}} \;\; C_{LB} \left(\mathbb{Q}^{[l,k]}, {\mathbb{Q}^{[l-1]}}^\star \right)  -   {\lambda^{[l,k-1]}}^\star P_{\text{tot}}\left(\mathbb{Q}^{[l,k]}\right) \;\;  \text{s.t.} \;\; \text{(\ref{problem_SEE_max_1_c})} ,  \label{Dinkelbach_SEE_1} \\
{\lambda^{[l,k]}}^\star & \leftarrow C_{LB,ab} \left({\mathbb{Q}^{[l,k]}}^\star, {\mathbb{Q}^{[l-1]}}^\star \right) / P_{\text{tot}}\left({\mathbb{Q}^{[l,k]}}^\star\right), \label{Dinkelbach_SEE_2}
\end{align}
associated with the $l$-th outer iteration and $k$-th inner iteration.   

It is observed that (\ref{Dinkelbach_SEE_1}) is a jointly convex problem over the optimization variables $\mathbb{Q}^{[k,l]}$ and \revPeterNew{can efficiently be implemented} via {the} MAX-DET algorithm \cite{vandenberghe1998determinant}, whereas (\ref{Dinkelbach_SEE_2}) can be obtained via direct evaluation. 
The defined algorithm steps, both outer and inner loop iterations, are continued until a jointly stable point is obtained, see Algorithm~\ref{SUIAP} for more details.    
\vspace{-2mm}
\subsubsection{Convergence}
Via the application of the \revPeterNew{Dinkelbach's algorithm} on the class of concave-over-affine fractional programs, the iterations of \revPeterNew{the} inner loop converge to a globally optimum solution for  (\ref{eq_SEE_lower_bound_SIA}). This follows from the strictly monotonic nature of the auxiliary function (\ref{eq_SEE_lower_bound_SIA}), and the convexity of the update (\ref{Dinkelbach_SEE_1}). For more information please see Appendix~\ref{appendix_dinkelbach} and the references therein. The following lemmas reveal the nature of the convergence at the outer loop. 
\begin{lemma} \label{lemma_SIA_Convergence}
(SIA sequence: \cite[Theorem~1]{marks1978technical}) Consider the optimization problem
\begin{align}  \label{lemma_SIA_convergence}   
\underset{ \ma{x} }{\text{min}} \;\; & g_0 (\ma{x})  \;\;  \text{s.t.} \;\; g_i(\ma{x}) \leq 0, \; \forall i \in \{1,\ldots,{I}\},                  
\end{align}     
where $g_i:\mathbb{R}^n \rightarrow \mathbb{R}^+$ and $g_i:\mathbb{R}^n \rightarrow \mathbb{R}$ are differentiable but potentially non-convex functions. Furthermore, consider the differentiable functions ($\forall i$) $\bar{g}_i$, approximating ${g}_i$ at $\ma{x}_0$, such that $(i) \; g_i(\ma{x}) \leq \bar{g}_i(\ma{x}, \ma{x}_0)$, $(ii) \; g_i(\ma{x}_0) = \bar{g}_i(\ma{x}_0, \ma{x}_0)$ and $(iii) \; \partial  g_i(\ma{x}_0) / \partial \ma{x} = \partial  \bar{g}_i(\ma{x}_0, \ma{x}_0) / \partial \ma{x}$. 
Then, upon the feasibility of an initial value $\ma{x}^{[0]}$, the sequence of approximate convex optimization problems 
\begin{align}  \label{lemma_SIA_convergence_approx}   
{\ma{x}^{[k]}}^{\star} \leftarrow \underset{\ma{x}^{[k]} }{\text{argmin}} \;\; & \bar{g}_0 \left(\ma{x}^{[k]}, {\ma{x}^{[k-1]}}^{\star}\right)  \;\;  \text{s.t.} \;\; \bar{g}_i \left(\ma{x}^{[k]}, {\ma{x}^{[k-1]}}^{\star}\right) \leq 0, \;\; i \in \{1,\ldots,{I}\},                  
\end{align}
\revPeterNew{converges} to a point satisfying the KKT conditions of the original problem (\ref{lemma_SIA_convergence}).  
\end{lemma}
\begin{proof}
The proof follows from two observations. Firstly, the sequence of $g_0 ({\ma{x}^{[k]}}^\star)$ leads to a necessary convergence. This is observed from the chain of inequalities $0 \leq g_0 ({\ma{x}^{[k]}}^\star) \leq g_0 ({\ma{x}^{[k-1]}}^\star) \leq \cdots \leq g_0 ({\ma{x}^{[0]}})$ as the optimal objective value is upperbounded by any feasible value. Secondly, the approximate problem (\ref{lemma_SIA_convergence_approx}) shares the same \revPeterNew{set of KKT conditions as (\ref{lemma_SIA_convergence}) at the point of  convergence, due to the stated properties $(i)-(iii)$}. The detailed proof of the latter case is articulated in~\cite{marks1978technical}, also see \cite{rev_1_2} for a similar discussion.
\end{proof}

\begin{lemma} \label{lemma_OuterLoopConvergence}
The properties $(i)-(iii)$ stated in Lemma~\ref{lemma_SIA_Convergence} hold for the approximation introduced in (\ref{eq_SEE_lower_bound_SIA}). 
\end{lemma}
\begin{proof}
The tightness $(ii)$, and the shared slope properties $(iii)$ are obtained directly by observing the nature of the inequality (\ref{eq_SEE_max_SEE_Taylor}), as the first-order Taylor approximation of the differentiable convex term $-\text{log}(x)$. The globally lower bound property $(iii)$ is observed since any first-order Taylor approximation of a convex function is also a global lower bound, see \cite[Subsection~3.1.3]{BV:04}. 
\end{proof}
The combination of Lemma~\ref{lemma_SIA_convergence} and Lemma~\ref{lemma_OuterLoopConvergence} conclude the convergence of the SUIAP algorithm to a point satisfying the KKT optimality conditions\footnote{\MinorRevOmid{However, due to the non-convex nature of the underlying problem, the global optimality of the converging point may not be theoretically guaranteed and the obtained solution depends on the used initialization. In Subsection~\ref{sec_AlgorithmAnalysis}, the optimal performance is numerically evaluated by repeating the SUIAP algorithm with several initializations. Although the optimality gap of the obtained solution \revOmid{may not} be theoretically guaranteed, it is observed that the proposed initialization leads to a negligible gap with the numerically obtained performance benchmark. \label{remarl_SUIAP_optimalitagap}}}.   
 \vspace{-2mm} \vspace{-0mm}
}

\subsubsection{Computational complexity}
The computational complexity of the algorithm is dominated by the steps of the determinant maximization in the inner loop. A general form of a MAX-DET problem is defined as 
\begin{align}
\underset{\ma{z}}{\text{min}} \;\; \ma{p}^T \ma{z} + \text{log}\left|{\ma{Y}(\ma{z})}^{-1}\right|, \;\; {\text{s.t.}} \;\; {\ma{Y}(\ma{z})} \succ 0,\; {\ma{F}(\ma{z})} \succeq 0, 
\end{align}  
where $\ma{z}\in \real^n$, and ${\ma{Y}(\ma{z})} \in \real^{n_Y \times n_Y} := \ma{Y}_0 + \sum_{i=1}^{n} {z_i \ma{Y}_i}$ and $ {\ma{F}(\ma{z})} \in \real^{n_F \times n_F} := \ma{F}_0 + \sum_{i=1}^{n} {z_i \ma{F}_i} $.
An upper bound to the computational complexity of the above problem is given as
\begin{align} \label{SUAIP_complexity}
\mathcal{O}\Big( \gamma_{\text{in}} \sqrt{n} \big(n^2 + n_Y^2 \big) n_F^2 \Big),
\end{align}
see \cite[Section~10]{vandenberghe1998determinant}. In our problem $n = 2N_{A}^2  + N_{B}^2 $ representing the dimension of real valued scalar variable space, and $n_Y = 2M_{B} + 2M_{E}$ and $n_F = 2N_{B} + 4N_{A}+ 2$, representing the dimension of the determinant operation and the constraints space, respectively.   \vspace{-1mm}
\begin{remark}
The above analysis intends to show how the bounds on computational complexity are related to different problem dimensions. Nevertheless, the computational load may vary in practice, depending on the implementation, the used numerical solver, and the number of optimization iterations required to obtain convergence. Please see Subsection~\ref{sec_AlgorithmAnalysis} for a numerical analysis on the algorithm computational complexity.
\end{remark}

\begin{algorithm}[H] 
{{ \begin{algorithmic}[1] 
\State{$l,k \leftarrow  {0}; \; \lambda^{[0,0]} \leftarrow {0}; \; \mathbb{Q}^{[0]}  \leftarrow  \text{Subsection~\ref{sec_SUIAP_init}}$} ;   \Comment{initialization}
\Repeat  \Comment{outer loop}
\State{$l \leftarrow  l + 1; \; {\lambda^{[0,l]}}^\star \leftarrow  {\lambda^{[k,l-1]}}^\star ; \; \mathbb{Q}^{[0,l]} \leftarrow  {\mathbb{Q}^{[k,l-1]}}^\star ; \;  k \leftarrow  0; $}
\Repeat \Comment{inner loop (Dinkelbach alg.)}
\State{$k \leftarrow  k + 1;$}
\State{$ \left\{ {\mathbb{Q}^{[l,k]} }^\star , {\lambda^{[l,k]}}^\star \right\} \leftarrow \text{(\ref{Dinkelbach_SEE_1}), (\ref{Dinkelbach_SEE_2})} ; $}
\State{$C \leftarrow  C_{LB,ab} \left({\mathbb{Q}^{[l,k]}}^\star , {\mathbb{Q}^{[l-1]}}^\star \right)  -   {\lambda^{[l,k-1]}}^\star P_{\text{tot}}\left({\mathbb{Q}^{[l,k]}}^\star \right);$}
\Until{$C \leq C_{\text{min}}$}
\Until{${\lambda^{[k,l]}}^\star  - {\lambda^{[0,l]}}^\star  \leq \lambda_{\text{min}} $}
\State{\Return$\left\{{\mathbb{Q}^{[k,l]}}^\star , {\lambda^{[k,l]}}^\star \right\}$}
  \end{algorithmic} }}
 \caption{\scriptsize{SUIAP algorithm for SEE maximization. $C_{\text{min}}$ ($\lambda_{\text{min}}$) represents the convergence threshold for outer (inner) iterations.} } \label{SUIAP}
\end{algorithm}   	 
 
%
%
  %
%

 

%% file: Sections/sec_SEE_max_BD.tex
In this part we study the case \revPeterNew{where bidirectional communication is established between Alice and Bob, such that both Alice and Bob are enabled with FD capabilities}. An FD bidirectional setup is interesting as it enables the usage of the same channel for both communication directions, and leads to a higher spectral efficiency \cite{DMBS:12}. Moreover, the jamming power at both, Alice and Bob, can be reused to improve security \revPeterNew{in} both directions\footnote{\revPeterNew{The reason for this is that} the jamming sent to Eve from each single node degrades Eves reception quality from both communication directions.}, and potentially improve the resulting SEE. However, the coexistence of all signal transmissions on a single channel results in a higher number of interference paths, which calls for a smart design regarding the signal and jamming transmit strategies at Alice and Bob. \par
In order to update the defined setup to a bidirectional one, we denote the number of receive antennas, and the self-interference channel at Alice as $M_{A}$ \revPeterNew{and} $\ma{H}_{aa}$, respectively. Moreover, we \MinorRevOmid{denote the} data transmission from Bob as $\ma{q}_b \sim \mathcal{CN} \left( \ma{0}_{N_{B}}, \ma{Q}_b \right)$. Following the same signal model for the transmission of data and jamming signals as in (\ref{eq_model_tx_alice})-(\ref{eq_model_e_rx_b}), \MinorRevOmid{the received interference-plus-noise covariance matrix at Eve, respectively associated with the Alice-Bob and Bob-Alice communications are expressed as 
\begin{align}                                   
\ma{\Sigma}_{e-a}^{\text{BD}} & = \ma{\Sigma}_e + \kappa_b \ma{H}_{be} \text{diag}\left( \ma{Q}_b\right) \ma{H}_{be}^H + \rho \ma{H}_{be} \ma{Q}_b \ma{H}_{be}^H,   \label{eq_BD_Sigma_update_Eve_a}  \\
\ma{\Sigma}_{e-b}^{\text{BD}} & = \ma{\Sigma}_e + \kappa_b \ma{H}_{be} \text{diag}\left( \ma{Q}_b\right) \ma{H}_{be}^H + \rho \ma{H}_{ae} \ma{Q}_a \ma{H}_{ae}^H,   \label{eq_BD_Sigma_update_Eve_b} 
\end{align}
where $\ma{\Sigma}_e$ is given \revPeterNew{in} (\ref{eq_model_Sigma_e}), and $\rho\in\{0,1\}$ is dependent on the decoding strategy at Eve\footnote{\MinorRevOmid{In particular, $\rho=1$ indicates the system where Eve is restricted to a linear reception strategy, similar to the other communication nodes, which represents a favorable (optimistic) scenario. On the other hand, $\rho=0$ indicates the system where Eve enjoys a successive decoding and interference cancellation capability, representing a worst-case (conservative) scenario, see \cite{6781609, 7463025}.}}. }
Similarly, the received interference-plus-noise signal covariance at Alice and Bob are respectively expressed as \MinorRevOmid{
\begin{align}                                   
 \ma{\Sigma}_a^{\text{BD}} &= \ma{H}_{ba}\ma{W}_b\ma{H}_{ba}^H + \sigma_{\text{n},a}^2 \ma{I}_{M_{A}} + \kappa_b \ma{H}_{ba} \text{diag} \left(\ma{Q}_b + \ma{W}_b\right) \ma{H}_{ba}^H + \kappa_a \ma{H}_{aa} \text{diag} \left( \ma{W}_a + \ma{Q}_a \right) \ma{H}_{aa}^H \nonumber \\ &   + \beta_a \text{diag} \Big( \ma{H}_{ba} \left( \ma{Q}_b + \ma{W}_b \right) \ma{H}_{ba}^H   + \ma{H}_{aa} \left( \ma{W}_a + \ma{Q}_a \right) \ma{H}_{aa}^H   + \sigma_{\text{n},a}^2 \ma{I}_{M_{A}} \Big) , \label{eq_BD_Sigma_Alice} \\
 \ma{\Sigma}_b^{\text{BD}} & = \ma{\Sigma}_b + \kappa_b \ma{H}_{bb} \text{diag}\left( \ma{Q}_b\right) \ma{H}_{bb}^H   +  \beta_b \text{diag}\left( \ma{H}_{bb}  \ma{Q}_b \ma{H}_{bb}^H \right),    \label{eq_BD_Sigma_update_Bob}  
\end{align}  }
where $\ma{\Sigma}_b$ is given in (\ref{eq_model_Sigma_b}), $\beta_a \in \real^+$ is the distortion coefficient for the reception at Alice, and $\sigma_{\text{n},a}^2$ represents the thermal noise variance at Alice. \MinorRevOmid{ The SEE for the defined BD system is then obtained as
\begin{align} \label{eq_BD_Sec_cap}                                  
\text{SEE}^{\text{BD}} = \frac{ \left\{\tilde{C}^{\text{BD}}_{ab} \right\}^{+} + \left\{\tilde{C}^{\text{BD}}_{ba} \right\}^{+}}{ P^{\text{BD}}_{A} + P^{\text{BD}}_{B} } ,
\end{align}
where 
\begin{align}
\tilde{C}^{\text{BD}}_{ab} =  \text{log} \Big|  &  \ma{I} +  \ma{H}_{ab}\ma{Q}_a\ma{H}_{ab}^H \left(\ma{\Sigma}_b^{\text{BD}}\right)^{-1} \Big| - \text{log} \Big|  \ma{I} 
 +  \ma{H}_{ae}\ma{Q}_a\ma{H}_{ae}^H \left(\ma{\Sigma}_{e-a}^{\text{BD}}\right)^{-1} \Big|, \label{C_BD_ab}\\
\tilde{C}^{\text{BD}}_{ba} =  \text{log} \Big|  &  \ma{I} +  \ma{H}_{ba}\ma{Q}_b\ma{H}_{ba}^H \left(\ma{\Sigma}_a^{\text{BD}}\right)^{-1} \Big| - \text{log} \Big|  \ma{I} 
 +  \ma{H}_{be}\ma{Q}_b\ma{H}_{be}^H \left(\ma{\Sigma}_{e-b}^{\text{BD}}\right)^{-1} \Big| \label{C_BD_ba}
\end{align}
are obtained following the same concept as in (\ref{eq_model_ab_cap}). Moreover, $P^{\text{BD}}_{B} := P_{B} +  \frac{1 + \kappa_b}{\mu_B} \text{tr}\left(\ma{Q}_b\right)$ and $P^{\text{BD}}_{A} := P_{A} + P_{\text{FD}}$, respectively represent the power consumption in the bidirectional system, where $P_A,P_B,P_{\text{FD}}$ are given in (\ref{eq_model_power_Alice}) and (\ref{eq_model_power_Bob}). In order to obtain a mathematically tractable framework, similar to (\ref{problem_SEE_max_1}), we resort to a relaxed version of the SEE by removing the non-linear operators $\left\{ \right\}^{+}$ from the formulation of (\ref{eq_BD_Sec_cap}). The relaxed optimization problem is hence formulated as 
\begin{align} \label{opt_SEEMax_BD_Relaxed}
\underset{ \ma{Q}_a,\ma{Q}_b,\ma{W}_a, \ma{W}_b \in \mathcal{H}}{\text{max}} \;\; & \text{SEE}_p^{\text{BD}} := \frac{ \tilde{C}^{\text{BD}}_{ab}  + \tilde{C}^{\text{BD}}_{ba} }{ P^{\text{BD}}_{A} + P^{\text{BD}}_{B} } \;\;\; \text{s.t.} \;\; P^{\text{BD}}_{A}  \leq P_{A,\text{max}}, \;\;   P^{\text{BD}}_{B} \leq P_{B,\text{max}}.
\end{align}
The following lemma explains the effectiveness of the employed relaxation.  
\begin{lemma} \label{lemma_BD_Positive_C}
Let $\left( \ma{Q}_a^{\star},\ma{Q}_b^{\star},\ma{W}_a^{\star}, \ma{W}_b^{\star} \right)$ {\MinorRR{be a KKT solution to (\ref{opt_SEEMax_BD_Relaxed}) for any chosen value of $\rho \in \{0,1\}$. Then, it holds 
\begin{align} \label{opt_SEEMax_BD_ZeroRelaxationGap}
\text{SEE}_p^{\text{BD}} \left( \ma{Q}_a^{\star},\ma{Q}_b^{\star},\ma{W}_a^{\star}, \ma{W}_b^{\star} \right) = \text{SEE}^{\text{BD}} \left( \ma{Q}_a^{\star},\ma{Q}_b^{\star},\ma{W}_a^{\star}, \ma{W}_b^{\star} \right),
\end{align}
which means that $\tilde{C}^{\text{BD}}_{ab}$ and $\tilde{C}^{\text{BD}}_{ba}$ will be non-negative for any globally or locally optimum solution to (\ref{opt_SEEMax_BD_Relaxed}), resulting in zero relaxation gap.  }}
%
%
\end{lemma}
\begin{proof}
See Appendix~\ref{appendix_lemma_BD_Positive_C}. 
\end{proof}  \vspace{0mm}  }
\subsection{Extended SUIAP for bidirectional-SEE maximization}
{\MinorRevOmid{By employing the results of Lemma~\ref{lemma_BD_Positive_C},}} it is observed that the SEE maximization problem (\ref{opt_SEEMax_BD_Relaxed}) shares a similar mathematical structure in relation to the transmit covariance matrices, i.e., $\ma{Q}_{\mathcal{X}}, \ma{W}_{\mathcal{X}}$, ${\mathcal{X}} \in \{a,b\}$ as addressed for (\ref{problem_SEE_max_1}). \revPeterNew{Hence, a procedure similar to \revPeter{the} SUIAP algorithm can be} {\MinorRR{employed to obtain an iterative solution, with a guaranteed convergence to a point satisfying KKT conditions.}} The computational complexity of each Dinkelbach step is obtained similar to (\ref{SUAIP_complexity}), where $n = 2N_{A}^2 + 2N_{B}^2 $, $n_Y = 2M_{B} + 2M_{A} + 2M_{E}$ and $n_F = 4N_{B} + 4N_{A}+2$.                 
  

%% file: Sections/sec_SSUM.tex
It is usually challenging to obtain an accurate estimate of $\ma{H}_{ae}$ and $\ma{H}_{be}$, due to the lack of collaboration from Eve and mobility. In this part, we consider the case where the channel matrices are known only partially, i.e., \revPeterNew{only statistical knowledge of $\ma{H}_{ae}, \ma{H}_{be}$ is available}\footnote{\revOmid{\revPeterNew{For a non-collaborative eavesdropper, statistical CSI can be obtained via blind channel estimation~\cite{CSI_BlindEstimation}, by channel estimation based on previously overheard transmissions of the eavesdropper, or via location-based estimation methods~\cite{Piropa_CSI} in case information on the environment or potential eavesdropper locations is known. For the scenarios where no CSI-related information can be extracted, the distribution will hold a uniform probability over all feasible channel values.}}}, considering a similar setup as defined in Section~\ref{sec:model}. \revPeterNew{Remark: In this section Alice is considered to be a HD node.} An optimization problem for maximizing the statistical expectation of the SEE is written as 
\begin{subequations} \label{problem_SSUM_1}
\begin{align}
\underset{\mathbb{Q}}{ \text{max}} \;\;   & \mathbb{E}_{\ma{H}_{ae}, \ma{H}_{be}} \left\{ {\text{SEE} \left( \mathbb{Q} \right) } \right\}   \label{problem_SSUM_1_a} \\
{\text{s.t.}} \;\;\; & \text{(\ref{problem_SEE_max_1_c})}.  
\end{align}  
\end{subequations}
It is worth mentioning that the consideration of statistical CSI on secrecy capacity with single antenna receivers is studied in \cite{lin2013secrecy}, considering a fast-fading channel case and extended for SEE maximization in \cite{zappone2016energy}, assuming HD operation of the nodes. In this work, we consider a more general case, \revPeter{where} the channel to Eve may be stationary, however not known accurately due to the lack of collaboration from Eve. In order to turn (\ref{problem_SSUM_1}) into a tractable form we write
\begin{align} \label{Eq_SSUM_SAA}
 \mathbb{E}_{\ma{H}_{ae}, \ma{H}_{be}} \left\{ {\text{SEE}} \right\}  = \frac{ \mathbb{E}_{\ma{H}_{ae}, \ma{H}_{be}} \left\{ \left\{ \tilde{C}_{ab} \left(\mathbb{Q}\right) \right\}^{+} \right\} }{ P_{\text{tot}} \left(\mathbb{Q} \right) } 
{\approx} \frac{ \frac{1}{|\mathbb{G}_C|} \sum_{i \in \mathbb{G}_C} \left\{ \tilde{C}_{\text{s},i} \left(\mathbb{Q} \right) \right\}^{+}  } {P_{\text{tot}} \left(\mathbb{Q} \right)} =: \text{SAA} \left( \mathbb{Q} \right),  
\end{align} 
where the latter is obtained via \revOmid{sample average} approximation (SAA) \cite{kim2015guide}, such that the equality holds \revPeter{for} $|\mathbb{G}_C| \rightarrow \infty$, $\mathbb{G}_C$ being the index set of the sampled channel realizations. Moreover, $\tilde{C}_{\text{s},i} \left(\mathbb{Q} \right) := \tilde{C}_{ab} \left(\mathbb{Q} , \ma{H}_{ae,i}, \ma{H}_{be,i} \right)$ where $\ma{H}_{ae,i}, \ma{H}_{be,i}$ represent the $i$-th realization of the channel matrices drawn from the given distribution\revPeter{s}. The approximated problem is hence expressed as 
\begin{align} \label{problem_SSUM_2}
\underset{\mathbb{Q}}{ \text{max}} \;\;\; \text{SAA}\left( \mathbb{Q} \right)\;\;\; {\text{s.t.}} \;\;\; \text{(\ref{problem_SEE_max_1_c})}.  
\end{align}  
\revOmid{Note that the above formulation is still challenging due to three reasons:} Firstly, while the application of SAA turns the statistical expectation into a linear sum for any arbitrary channel distribution, it results in a high computational complexity as $|\mathbb{G}_C|$ increases. This calls for a smart choice of $|\mathbb{G}_C|$, compromising accuracy with algorithm complexity. Secondly, unlike the scenario with perfect CSI and also the case presented in \cite{lin2013secrecy, zappone2016energy} considering a fast fading channel situation, the $\{\}^{+}$ operation may not be ignored. The reason for this is that some of the channel realizations may result in a negative $\tilde{C}_{\text{s}}$, while the statistical expectation remains effectively positive. And \revPeterNew{thirdly}, similar to the studied problem in (\ref{problem_SEE_max_1}), the above objective presents a non-concave over affine fractional program which is not tractable in general.

\subsection{Successive selection and statistical lower bound maximization (SSSLM)}
In order to address the aforementioned challenges, we propose a successive selection and statistical lower bound maximization (SSSLM) algorithm, which converges to a stationary point of (\ref{problem_SSUM_2})\footnote{\revOmid{Please note that in contrast to Subsection~\ref{subsec_SUIAP}, the operating objective in this part is not a differentiable one, \revPeterNew{hence it violates the conditions given by SIA \cite{marks1978technical}.} In this regard we follow a variation of SIA, i.e., the successive upper-bound minimization (SUM) method \cite{razaviyayn2013unified}, generalizing the convergence arguments in SIA-based methods for non-smooth problems. The proposed SSSLM algorithm is composed of three nested loops: Separation of the SAA into smooth and non-smooth parts at the outer loop, construction of an effective lower bound to SAA as the intermediate loop, and maximization of the constructed bound in the inner loop.}}.
A detailed description of the algorithm steps is given in the following.

\subsubsection{Initialization}  \label{subsec_SSUM_init}
The algorithm starts by generating the channel instances $\ma{H}_{ae,i}, \ma{H}_{be,i} , \forall i \in \mathbb{G}_C$, drawn from the known statistical distribution of the channels. \revOmid{The number of channel realizations, i.e., $\left| \mathbb{G}_C \right|$, should be chosen large enough to capture the channel statistics in SAA with adequate accuracy, however, should be kept small to reduce computational complexity. The analytical expression for the choice of $\left| \mathbb{G}_C \right|$ is given in \cite[Theorem 5.18]{ruszczynski2003stochastic}, depending on the required statistical accuracy and the given probability distribution.} For the initialization of $\mathbb{Q}$, we follow the approximation  
\begin{align}
\mathbb{E}_{\ma{H}_{ae}, \ma{H}_{be}} \left\{ {\text{SEE}} \left(\mathbb{Q},  \ma{H}_{ae}, \ma{H}_{be} \right)  \right\} \approx {\text{SEE}} \left(\mathbb{Q},  \mathbb{E} \left\{ \ma{H}_{ae} \right\}, \mathbb{E}\left\{ \ma{H}_{be}\right\} \right), 
\end{align}
where the expectations $\mathbb{E} \left\{ \ma{H}_{ae} \right\}, \mathbb{E} \left\{ \ma{H}_{be} \right\}$ are obtained from the statistical distribution of the channels. Note that the right side of the approximation corresponds to the objective addressed in Subsection~\ref{subsec_SUIAP}, where the \revOmid{SUIAP algorithm} is applied. The obtained solution from \revOmid{SUIAP} is then used as an initialization to the SSSLM algorithm. \vspace{-2mm}
\subsubsection{Outer loop}
In each outer iteration, the objective is decomposed as 
\begin{align}  \label{SSUM_decomposition}
\text{SAA} \left( \mathbb{Q} \right) =  \frac{\sum_{i \in \mathbb{G}_{{C_1}}} \left\{\tilde{C}_{\text{s},i} \left(\mathbb{Q} \right) \right\}^+  + \sum_{i \in \mathbb{G}_{{C_2}}} \left\{ \tilde{C}_{\text{s},i} \left(\mathbb{Q} \right) \right\}^+  }{|\mathbb{G}_C| P_{\text{tot}} \left(\mathbb{Q} \right) }
\end{align}
by separating the set of channel realizations into the disjoint sets $\mathbb{G}_{{C_1}}$ and $\mathbb{G}_{{C_2}}$, such that $\mathbb{G}_{{C}} = \mathbb{G}_{{C_1}} \cup \mathbb{G}_{{C_2}}$. In particular, the set $\mathbb{G}_{{C_1}}$ is updated in each outer iteration as 
\begin{align}  \label{SSUM_def_F_C_1}
\mathbb{G}_{{C_1}}^{(\text{new})} \leftarrow \left\{  \forall i \;\; \vert \;\; i \in {\mathbb{G}_{{C_1}}} \;\;  \text{or}  \;\;  \tilde{C}_{\text{s},i} \left( \mathbb{Q} \right) = 0  \right\}, 
\end{align} 
\revOmid{where $\mathbb{Q}$ is given from the last intermediate loop, and results in the separation of smooth and non-smooth parts of the objective in (\ref{SSUM_decomposition})}. The algorithm converges when the constructed set ${\mathbb{G}_{{C_1}}}$ does not change. As it will be elaborated, the set members in $\mathbb{G}_{{C_1}}$ incur a high computational complexity, but are capable of resolving the non-smooth points by maintaining the same directional derivative as the SAA. On the other hand, the set members in $\mathbb{G}_{{C_2}}$ are resolved with lower computational complexity, however, they are not capable of handling non-smooth situations.  \vspace{-2mm}
%
\subsubsection{Intermediate loop}
In each intermediate iteration a lower bound is constructed from the original objective SAA, namely $\text{SAA}_{{LB}}$, using the value of $\mathbb{Q}$ from the last inner loop, i.e., $\mathbb{Q}_0$. In order to construct $\text{SAA}_{{LB}}$ we undertake three steps. Firstly, the operator $\{\}^{+}$ is removed from SAA for $i \in \mathbb{G}_{{C_2}}$, which results in a global lower bound. Secondly, concave and tight lower bounds of the functions $\tilde{C}_{\text{s},i}$ are constructed at the point $\mathbb{Q}_0$, denoted as $\hat{C}_{\text{s},i} \left(\mathbb{Q} , \mathbb{Q}_0 \right)$, by applying the inequality (\ref{eq_SEE_max_SEE_Taylor}) on the convex parts. Please note that the value of $\tilde{C}_{\text{s},i}$ functions may be negative at $\mathbb{Q}_0$ for some $i\in \mathbb{G}_{C_2}$, resulting in a bias to the original objective. In order to obtain a tight lower bound, we define the set
\begin{align}  \label{SSUM_set_def_C_2_+}
{\mathbb{G}_{{C}_2^+}} := \left\{  \forall i \;\; \vert \;\; i \in {\mathbb{G}_{{C_2}}} , \; \tilde{C}_{\text{s},i} \left( \mathbb{Q}_0 \right) \geq 0  \right\} ,
\end{align}  
representing the subset of channel realizations resulting in a non-negative $\tilde{C}_{\text{s},i}$ at $\mathbb{Q}_0$. The corresponding lower bound function is then obtained as
\begin{align} \label{SSUM_SAA_LB} 
& \text{SAA}_{{LB}} \left(\mathbb{Q} , \mathbb{Q}_0 \right) := \frac{ \sum_{i \in \mathbb{G}_{{C_1}}}  \left\{ \hat{C}_{\text{s},i} \left(\mathbb{Q} , \mathbb{Q}_0 \right) \right\}^+  + \sum_{i \in \mathbb{G}_{{C_2^+}}} \hat{C}_{\text{s},i} \left(\mathbb{Q} , \mathbb{Q}_0 \right) }{ |\mathbb{G}_{{C}}| P_{\text{tot}} \left(\mathbb{Q} \right) }.
\end{align}  
It can be verified that the constructed lower bound is tight at the point of approximation, i.e., $\text{SAA} \left(\mathbb{Q}_0\right) = \text{SAA}_{LB} \left(\mathbb{Q}_0, \mathbb{Q}_0\right)$, see Appendix~\ref{appendix_Lemma_SSUM_convergence}. The obtained lower bound is then optimally maximized in the inner loop. The iterations of the intermediate loop converge when $\mathbb{Q}_0$, and hence $\text{SAA}_{{LB}}$, (almost) does not change in subsequent intermediate iterations. 
\vspace{-2mm}
\subsubsection{Inner loop}
The inner loop is dedicated to optimally maximize $\text{SAA}_{{LB}}$, under the original problem constrains (\ref{problem_SSUM_2}). Note that the $\text{SAA}_{{LB}}$ is not tractable in the current form, due to the $\{\}^+$ operation. In order to obtain the optimum solution we equivalently write the maximization problem in the inner loop as 
\begin{align} \label{SSUM_inner_SAA_bar}
\underset{ a_i \in\{0,1\}, i \in \mathbb{G}_{C_1}}{\text{max}} \;\; \underset{ \mathbb{Q}}{\text{max}} \;\;  & \overbar{\text{SAA}_{{LB}}}, \;\;\; {\text{s.t.}} \;\;\; \text{(\ref{problem_SEE_max_1_c})},
\end{align}
where $\overbar{\text{SAA}_{{LB}}}$ is obtained by replacing the terms $\left\{ \hat{C}_{\text{s},i}\right\}^+$ in (\ref{SSUM_SAA_LB}) by $a_i \hat{C}_{\text{s},i}$. Please note that for fixed values of $a_i$, $i\in \mathbb{G}_{C_1}$, the function $\overbar{\text{SAA}_{{LB}}}$ is a concave over affine fraction, and can be maximized to optimality via the application of the Dinkelbach algorithm. Hence (\ref{SSUM_inner_SAA_bar}) can be solved by repeating the Dinkelbach algorithm for all $2^{|\mathbb{G}_{C_1}|}$ possible combinations of $a_i$, $i\in \mathbb{G}_{C_1}$, however, requiring a large number of \revPeter{Dinkelbach} iterations. The optimization problem corresponding to the $k$-th inner iteration is expressed as     

\begin{subequations} \label{problem_SSUM_max_Dink}
\begin{align}
\underset{\ma{a}^{[k]} \in \mathbb{A}^{[k]} }{ \text{max}} \underset{\mathbb{Q}^{[k]}}{ \text{max}} \;\;   & \overbar{{\text{SAA}}_{\overline{LB}}} \left(\mathbb{Q}^{[k]} , \mathbb{Q}^{[0]}, \ma{a}^{[k]} \right) -  \lambda^{[k-1]}   P_{\text{tot}} \left(\mathbb{Q}^{[k]} \right) \label{problem_SSUM_max_Dink_a} \\
{\text{s.t.}} \;\;\; & \text{(\ref{problem_SEE_max_1_c})}. 
\end{align}  
\end{subequations}
where $\overbar{{\text{SAA}}_{\overline{LB}}}$ is \revOmid{the numerator} in $\overbar{\text{SAA}_{{LB}}}$, and $\mathbb{Q}^{[0]}$ is the point for the construction of ${\text{SAA}}_{{LB}}$, given from the intermediate loop. Moreover, the vector $\ma{a} \in \{0,1\}^{|\mathbb{G}_{C_1}|}$ stacks the values of $a_i, \forall i \in \mathbb{G}_{C_1}$, and $\mathbb{A}^{[k]} \subset \{0,1\}^{|\mathbb{G}_{C_1}|}$. It is observed that for a given $\ma{a}^{[k]}, \lambda^{[k-1]}$, (\ref{problem_SSUM_max_Dink}) is a jointly convex optimization problem, and is solved to optimality via MAX-DET algorithm \cite{vandenberghe1998determinant}. Hence, the optimum $\ma{a}^{[k]}, \mathbb{Q}^{[k]}$ are obtained by repeating the MAX-DET algorithm for all combinations $\ma{a}^{[k]} \in \mathbb{A}^{[k]}$. The value of $\lambda$ is then updated by applying the obtained $\mathbb{Q}^{[k]},  \ma{a}^{[k]}$ as 
\begin{align}\label{problem_SSUM_max_Dink_lambda}
\lambda^{[k]} = \overbar{\text{SAA}_{\overline{LB}}} \left(\mathbb{Q}^{[k]} , \mathbb{Q}^{[0]}, \ma{a}^{[k]} \right)/ P_{\text{tot}} \left(\mathbb{Q}^{[k]} \right). 
\end{align}
Please note that the set $\mathbb{A}^{[k]}$, is initialized as $\{0,1\}^{|\mathbb{G}_{C_1}|}$ but is reduced in each iteration. The following lemma clarifies this reduction. \vspace{-3mm}
\begin{lemma} \label{lemma_SSUM_ModifiedDinkelbach}
Let $g_k (\ma{a}_0)$ be the optimal value of the objective (\ref{problem_SSUM_max_Dink}) at inner iteration $k$, for the given combination $\ma{a}^{[k]} = \ma{a}_0$. 
Then, if $g_k(\ma{a}_0)$ is negative, then the combination ${\ma{a}}_0$ will not be an optimum combination.  
\end{lemma} \vspace{-3mm}
\begin{proof}
Due to the monotonic improvement of $\lambda$ in every iteration, and the fact that $P_{\text{tot}} \geq 0$, the value of $g_k (\ma{a}_0)$ will never improve after further iterations. This also results in a negative value of $g_k (\ma{a}_0)$ at the optimality. Since at least one of the combinations $\ma{a} \in \{0,1\}^{|\mathbb{G}_{C_1}|}$ equalizes the objective to zero at the optimality, the combination ${\ma{a}_0}$ will never be optimal.  
\end{proof} 
As a result of Lemma~\ref{lemma_SSUM_ModifiedDinkelbach}, once a combination $\ma{a}_0$ results in a negative value of the objective, then it is safely removed from $\mathbb{A}$ for the next iteration, see Algorithm~\ref{alg_SSSLM}. Note that the above process reduces the required computational complexity, compared to the separately applying the Dinkelbach method on all combinations, in two ways. Firstly, the parameter $\lambda$ is only updated jointly, for all combinations $\ma{a} \in \mathbb{A}$. Secondly, the monotonic reduction in $|\mathbb{A}|$ in each iteration, results in a smaller computational demand in finding the solution to (\ref{problem_SSUM_max_Dink}). 
  \vspace{-1mm}
\subsubsection{Convergence}
The proposed SSSLM algorithm converges to a stationary point of the original optimization problem (\ref{problem_SSUM_2}). In order to observe this, we first verify the convergence of the algorithm. Afterwards, we show that the converging point is a stationary point of (\ref{problem_SSUM_2}). \par
It is observed that the constructed lower bound in each step of the intermediate loop is maximized to the optimality via the application of the modified Dinkelbach algorithm. On the other hand, the value of ${\text{SAA}}_{{LB}} ( \mathbb{Q})$ after the construction of the new lower bound in each intermediate iteration experiences an improvement. This is grounded on the re-calculation of $\hat{C}_{\text{s},i}$ at the point of approximation, and elimination of the channel instances from $\mathbb{G}_{C_2}$ which result in a negative $\tilde{C}_{\text{s},i}$. Since the both the aforementioned updates result in a monotonic improvement of ${\text{SAA}}_{{LB}} ( \mathbb{Q})$ and as the $\text{SAA}$ is bounded from above, the iterations of inner and intermediate loop will result in a necessary convergence. The convergence of the intermediate loop subsequently results in the necessary convergence of the outer loop, due to the monotonic increase of $|\mathbb{G}_{C_1}|$ after each outer iteration, and the fact that $|\mathbb{G}_{C_1}| \leq |\mathbb{G}_{C}|$. \par
In order to argue the properties of the converging point on the original objective, we observe that neither the SAA nor ${\text{SAA}}_{{LB}} ( \mathbb{Q})$ are necessarily differentiable at the point of convergence. This invalidates the convergence arguments used for \revOmid{SUIAP algorithm} from \cite{marks1978technical}. In this regard, we follow the guidelines given by the SUM method \cite{razaviyayn2013unified}, generalizing the convergence arguments in SIA-based methods for non-smooth problems. \vspace{-2mm}
\begin{lemma} \label{lemma_SSUM_Convergence}
Let $\mathbb{Q}^\star$ be a solution of SSSLM. Then the function SAA, i.e., original problem objective, and ${\text{SAA}}_{{LB}}$, i.e., the constructed lower bound at the last intermediate iteration, are tight and share the same directional derivatives at $\mathbb{Q}^\star$.
\end{lemma} \vspace{-2mm}
\begin{proof}
See Appendix~\ref{appendix_Lemma_SSUM_convergence}.
\end{proof} 
The results of Lemma~\ref{lemma_SSUM_Convergence}, together with the fact that ${\text{SAA}}_{{LB}} (\mathbb{Q}) \leq {\text{SAA}}(\mathbb{Q})$ for any feasible $\mathbb{Q}$, jointly satisfy the required assumption set \cite[Assumption~1]{razaviyayn2013unified}, and guarantee that the obtained converging point is indeed a stationary point of the original problem.  \par
\vspace{-1mm}
\begin{remark}
\MinorRevOmid{Similar to the SUIAP algorithm, the \revPeterNew{global} optimality of the obtained stationary point via the SSSLM algorithm may not be guaranteed, and the obtained solution depends on the used initialization, see Remark~\ref{remarl_SUIAP_optimalitagap}. In Subsection~\ref{sec_AlgorithmAnalysis}, the optimal performance is numerically evaluated by repeating the SSSLM algorithm with several initializations, where the average $1-3\%$ relative gap is observed for the proposed initialization in Subsection~\ref{subsec_SSUM_init}. } 
\end{remark}\vspace{0mm}
\subsubsection{Computational complexity}
The computational complexity of the algorithm is dominated by the maximization defined in (\ref{problem_SSUM_max_Dink}), solved via the MAX-DET algorithm in each inner iteration. The associated arithmetic complexity\footnote{The resulting computational complexity is dominated by the SAA sample size, due to the construction of (\ref{Eq_SSUM_SAA}), as well as the iterations of \revPeterNew{the} outer loop. In order to reduce the resulting computational efforts, the algorithm can be customized to a specific CSI {statistic}, e.g., by substituting the SAA a more efficient structure. \MinorRR{In this case, the achievable SEE must be approximated for a specific statistics in a tractable form, which is then used for the purpose of performance optimization in the design algorithm, e.g., see \cite{Ali_FD_FastFading} for a similar approach with Gaussian channels but for a different system objective. Another approach is to eliminate the operations in the outer loop when occurrence of the non-smooth points is not frequent. Moreover, the obtained initialization point in Subsection~\ref{subsec_SSUM_init} can serve as an intuitive low complexity solution. }  } is hence upper-bounded similar to (\ref{SUAIP_complexity}), where $\gamma_{\text{in}} \propto 2^{|\mathbb{G}_{C_1}|}$, $n = 2N_{A}^2  + N_{B}^2 $, $n_Y = |\mathbb{G}_{C}|(2M_{B} + 2M_{E})$, $n_F = 4N_{A}+2N_{B}+2$ \vspace{-5mm}
\begin{algorithm} 
 \caption{{SEE maximization using statistical CSI, via successive selection and statistical lower bound maximization (SSSLM). $C_{\text{min}}$ ($\lambda_{\text{min}}$) represents the convergence threshold for the intermediate (inner) iterations. } } \label{alg_SSSLM}  
 {\scriptsize{	\begin{algorithmic}[1] 
\State{$k,l,m,\lambda^{[0,0,0]} \gets  {0} ; \mathbb{G}_{C_1}^{[0]} \gets  {\emptyset};\; \mathbb{G}_{C}, \mathbb{Q}^{[0,0,0]}  \gets  \text{Subsection~\ref{subsec_SSUM_init}} ;$}  \Comment{initialize} 
\Repeat  \Comment{outer loop} 
\State{$m \leftarrow  m + 1 ; \; \lambda^{[0,0,m]} \leftarrow  \lambda^{[k,l,m-1]}, \mathbb{Q}^{[0,0,m]} \leftarrow  \mathbb{Q}^{[k,l,m-1]}; \; \mathbb{G}_{C_1}^{[m]} \leftarrow  \text{(\ref{SSUM_def_F_C_1})} ; \; l \leftarrow  0;$} 
\Repeat \Comment{intermediate loop}
\State{$l \leftarrow  l + 1; \; {\mathbb{G}_{{C}_2^+}} \gets \text{(\ref{SSUM_set_def_C_2_+})} ;  \text{SAA}_{{LB}} \gets  \text{(\ref{SSUM_SAA_LB})} ;$}
\Repeat \Comment{inner loop }
\State{$k \leftarrow  k + 1;\; \left\{\mathbb{Q}^{[k,l,m]}, \lambda^{[k,l,m]} \right\}  \gets \text{Dinkelbach's alg. (\ref{problem_SSUM_max_Dink})-(\ref{problem_SSUM_max_Dink_lambda})};$}
\Until{$\text{(\ref{problem_SSUM_max_Dink_a})} \leq C_{\text{min}}$}
\Until{$\lambda^{[k,l,m]} - \lambda^{[0,l,m]} \leq \lambda_{\text{min}} $}
\Until{$\mathbb{G}_{C_{1}}^{[m]} = \mathbb{G}_{C_{1}}^{[m-1]} $}
\State{\Return$\left\{\mathbb{Q}^{[k,l,m]}, \lambda^{[k,l,m]}\right\}$}
  \end{algorithmic} }} \vspace{-1mm} 
\end{algorithm} \vspace{-0mm}

    			

%% file: Sections/sec_simulations.tex
\input{./Sections/sim_alg_analysis}

In this section the performance of the studied MIMOME system is evaluated in terms of the resulting SEE, via numerical simulations. In particular, we are interested in a comparison between the performance of \revPeterNew{a} FD-enabled setup, compared to the case where all nodes operate in HD mode. Moreover, the evaluation of the proposed SEE-specific designs is of interest, in comparison to the available designs which target the maximization of the system's secrecy capacity. We assume that all communication channels follow an uncorrelated Rayleigh flat-fading model with variance $\rho_{\mathcal{X}} = \bar{\rho} / d_{\mathcal{X}}^{2}$, where $d_{\mathcal{X}}$ is the link distance and depends on the simulated geometry, $\mathcal{X} \in \{ab, ba, ae, be\}$. Moreover, in case that only \revPeterNew{statistical CSI} is available for the channels to Eve, we assume that $\ma{H}_{\mathcal{X}} = \sqrt{\rho_{\mathcal{X}}} \left( \ma{D}_{\mathcal{X}} \bar{\ma{H}}_{\mathcal{X}} +  \tilde{\ma{H}}_{\mathcal{X}} \right)$, ${\mathcal{X}} \in\{ae,be\}$, where $\sqrt{\rho_{\mathcal{X}}}\tilde{\ma{H}}_{\mathcal{X}}$ is the known channel estimate following similar statistics as for the exact CSI case, and $\sqrt{\rho_{\mathcal{X}}} \ma{D}_{\mathcal{X}} \bar{\ma{H}}_{\mathcal{X}}$ is the estimation error where $\ma{D}$ enforces the receive-side spatial correlation, and $\bar{\ma{H}}_{\mathcal{X}}$ includes i.i.d Gaussian elements with unit variance. For the self-interference channels we follow the characterization reported in \cite{FD_ExperimentDrivenCharact}. In this respect we have $\ma{H}_{bb} \sim \mathcal{CN}\left( \sqrt{\frac{\rho_{\text{si}} K_R}{1+K_R}} \ma{H}_0 , \frac{\rho_{\text{si}}}{1+K_R} \ma{I}_{M_{B}} \otimes \ma{I}_{N_{B}} \right)$, where $\rho_{\text{si}}$ represents the self-interference channel strength, $\ma{H}_0$ is a deterministic term\footnote{For simplicity, we choose $\ma{H}_0$ as a matrix of all-$1$ elements.}, and $K_R$ is the Rician coefficient. The statistics of the self-interference channel on Alice, i.e., $\ma{H}_{aa}$, is defined similarly. The resulting \revPeterNew{system's} SEE is evaluated by employing different design strategies, and averaged over $100$ channel realizations. Unless stated otherwise, the default simulated setup is defined as follows: $P_{\text{max}}:= P_{\mathcal{X},\text{max}} = 0\text{dB}$,  $P_0:= P_{\mathcal{X},0}= -20\text{dB}$, $\mu:= \mu_{\mathcal{X}} = 0.9$, $\kappa:= \kappa_{\mathcal{X}} = \beta_{\mathcal{X}} = -40\text{dB}$, $N:= N_{\mathcal{X}} = M_\mathcal{X} = 4$, $\mathcal{X},\in \{A,B\}$. Moreover we set $P_{\text{FD}}= 0$\footnote{This corresponds to the cancellation methods using a joint passive and baseband techniques, which do not implement auxiliary transmissions or additional analog circuitry. The impact of $P_{\text{FD}}$ is later studied in Fig.~\ref{fig_SUAIP}, addressing a general cancellation method.}, $\rho_{\text{si}} = 0\text{dB}$,  $K_R = 10$, $\bar{\rho}= -20\text{dB}$, and $\sigma_{\text{n}}^2 := \sigma_{\text{na}}^2 =$ $\sigma_{\text{nb}}^2 = \sigma_{\text{ne}}^2 = -40\text{dB}$. Three nodes are equidistantly positioned, with the distance equal to one\footnote{\revOmid{We consider unit-less parameters to preserve a general framework. However, the obtained SEE values can be interpreted as the number of securely communicated bits per-Hz per-Joule, assuming the power values are in Watt.}}. \par 
\vspace{-0mm}
\subsection{Algorithm analysis}  \label{sec_AlgorithmAnalysis}
Due to the iterative structure of the proposed algorithms and the possibility of local optimum points, the convergence behavior of the algorithms are of high interest, both as a verification for algorithm operation as well as an indication of the required computational effort. In this part, the performance of the SUAIP and SSSLM algorithms are studied in terms of the average convergence behavior and computational complexity. Moreover, the impact of the choice of the algorithm initialization is evaluated.


In Fig.~\ref{fig_alg_analysis} (a), the average convergence behavior of the SUAIP algorithm is depicted. As expected, a monotonic objective improvement is observed, with convergence in $5$-$20$ total outer iterations. 

In Figs.~\ref{fig_alg_analysis} \revPeterNewNew{(b)-(c), the impact of the proposed initializations for the SUAIP and SSSLM algorithms are depicted. In Fig.~\ref{fig_alg_analysis} (b) it is observed that for the SUIAP algorithm, the proposed initialization in Subsection~\ref{sec_SUIAP_init} reaches close to the benchmark performance\footnote{The benchmark performance is obtained by repeating the algorithm with several random initializations, and choosing the highest obtained SEE.}. For the SSSLM algorithm, the situation is prone to more randomness. This is since, in addition to the choice of the algorithm initialization, the solution is dependent on the used channel realizations used in the construction of SAA, see (\ref{Eq_SSUM_SAA}). In this regard, the resulting cumulative distribution function (CDF) of the resulting SEE values is depicted in Fig.~\ref{fig_alg_analysis}~(c), by examining $100$ instances of the SSSLM algorithm. It is observed that the resulting average SEE differs for different solution instances, however, within $2-3\%$ of the relative difference. This value is smaller for a system with HD nodes, due to the absence of FD jamming and the impact of residual self-interference which result in a simpler problem structure.}
\begin{table}[!t] 
\centering
	\renewcommand{\arraystretch}{1.1}
  	\caption{Average CPU Time}\label{tab_net_params}
  	\centering  \vspace{-3mm}
  	\begin{tabular}[t]{||c|c|c|c|c||}
	 	\hline
   		Algorithm & SUIAP-HD & SUIAP-FD & SSSLM-HD & SSSLM-FD   \\
   		\hline
   		\begin{tabular}{c} CPU Time [s] \\ (Initialization) \end{tabular}  & \begin{tabular}{c}  17.4  \\  (1.2)  \end{tabular} &  \begin{tabular}{c}  31.1  \\  (1.5)  \end{tabular} & \begin{tabular}{c}  $1.39\times 10^4$  \\  (17.7)   \end{tabular} &  \begin{tabular}{c}  $6.8\times 10^4$  \\  (32.2)  \end{tabular}  \\ \hline   
		\end{tabular} \vspace{-4mm}		 
\end{table} 
{The required average CPU time\footnote{The reported CPU time is obtained using an \revOmid{Intel Core i$5$ $3320$M processor with the clock rate of $2.6$ GHz and $8$ GB of random-access memory (RAM). As software platform we have used CVX~\cite{YGB:08}, together with MATLAB $2013$a on a $64$-bit operating system.}}, for SSSLM as well as the SUIAP algorithms are depicted in Table~\ref{tab_net_params}, applied on HD and FD scenarios. \revPeterNewNew{Moreover, the CPU time associated with the proposed initialization methods, which can be considered as an intuitive but sub-optimal but practical algorithm in each case, are given in parenthesis in the second row. It is observed that a design with FD-enabled jamming results in a higher CPU time, due to the additional problem complexity associated with the choice of jamming strategy and residual self-interference.}}

\vspace{0mm}
\subsection{Performance comparison} \label{Sim_benchmarks}
In this part the SEE performance of the FD-enabled system is evaluated via the application of the proposed SUAIP and SSSLM algorithms, and under different system conditions. In particular, we are interested in a comparison between the performance of an FD-enabled setup, compared to the case where all nodes operate in HD mode. Moreover, the evaluation of the proposed SEE-specific designs is of interest, in comparison to the available designs which target the maximization of the system's secrecy capacity. The following benchmarks are hence implemented to provide a meaningful comparison.   
\MinorRevOmid{\begin{itemize}[leftmargin=*]
\item \textit{SEE-FD}:~The proposed SUAIP (SSSLM) algorithm is implemented using the exact (statistical) CSI, where Bob is capable of FD operation. 
\item \textit{SEE-HD}:~Similar to \textit{SEE-FD}, but with restricting the operation of the nodes to the HD mode.   
\item \textit{CS-FD}:~The design with the intention of maximizing secrecy capacity. Bob is capable of FD operation.  
\item \textit{CS-HD}:~Similar to \textit{CS-FD}, but with restricting the operation of the nodes to the HD mode.
\end{itemize} }

\vspace{-0mm}
\subsubsection{FD-enabled jamming with exact CSI}  
\input{./Sections/sim_SUAIP}
In Figs.~\ref{fig_SUAIP}~(a)-(h) the average SEE performance of the defined benchmarks are evaluated, assuming availability of perfect CSI and FD operation at Bob. Hence, both Alice and Bob are simultaneously capable of transmitting AN, see Fig.~\ref{fig_model}.

In Fig.~\ref{fig_SUAIP}~(a) the impact of thermal noise variance is depicted. It is observed that a higher $\sigma_{\text{n}}^2$ results in a smaller SEE both for FD and HD setups. Moreover, a marginal gain for FD setup is obtained compared to the HD setup, if the noise variance is low. This is expected, since FD jamming becomes less effective when Eve is already distorted with high thermal noise power.

In Fig.~\ref{fig_SUAIP}~(b) the impact of the available transmit power budget ($P_{\text{max}}$) for each transceiver is depicted. It is observed that for small values of $P_{\text{max}}$, the resulting SEE is monotonically increasing with an increase in $P_{\text{max}}$. Moreover, the performance of the benchmark algorithms essentially converge for small values of $P_{\text{max}}$. This is grounded in the fact that for a system with low SNR condition, the positive impact of FD jamming disappears as observed from Fig.~\ref{fig_SUAIP}~(a). Conversely, for large values of $P_{\text{max}}$, the traditional designs result in a rapid decrease of SEE, where the proposed SUAIP method converges to a constant value. This is expected, since the designs with the target of maximizing the secrecy rate utilize the maximum available power budget, resulting in a sever degradation of SEE. Moreover, a visible gain is observed with the application of an FD jammer for a high $P_{\text{max}}$ region. \revPeter{Due to a high $P_{\text{max}}$, the link from Alice to Eve} also enjoys a higher SNR, which justifies the application \revPeter{of a FD jammer}.

In Fig.~\ref{fig_SUAIP}~(c) the impact of transceiver accuracy is depicted. As expected, a higher value of $\kappa$ results in a smaller achievable SEE, both in HD and FD setups. Moreover, it is observed that FD jamming can be beneficial only for a system with an accurate hardware operation, due to the impact of residual self-interference. However, results show that targeting SEE as the design objective results in a significant energy efficiency gain, compared to the traditional designs which target the maximization of secrecy rate.

In Fig.~\ref{fig_SUAIP}~(d) the impact of Eve's distance to Alice ($d_E$) is depicted. It is assumed that three nodes are positioned in a line with a total Alice-Bob distance of $100$, where Eve is positioned in between. It is observed that the system's SEE increases as $d_E$ increase, and Eve gets closer to \revPeter{Bob}. Moreover, the application of FD jamming becomes beneficial only when Eve is located in a close distance to Bob, and hence the channel between Bob and Eve, i.e., the jamming channel, is strong.    

In Figs.~\ref{fig_SUAIP}~(e) the impact of the number of antenna elements at Eve ($M_E$) on SEE is depicted. As expected, a larger $M_E$ results in a reduced SEE as it results in a stronger Alice-Eve channel. Moreover, the application of an FD jammer becomes gainful for a higher values of $M_E$, in order to counteract the improved Eve reception capability. 
   
In Figs.~\ref{fig_SUAIP}~(f)-(h), the impact of the transceiver's power efficiency is evaluated on the resulting system SEE. In particular, the impact of the zero-state power consumption ($P_{0}$) and PA efficiency ($\mu$) are depicted respectively in Figs.~\ref{fig_SUAIP}~(f) and (g). The impact of the additional power consumption for SIC ($P_{\text{FD}}$) on the system SEE is depicted for different noise regimes in Fig.~\ref{fig_SUAIP}~(h), \revPeterNewNew{where the two constant red lines represent the SEE for the HD setup.} It is observed that higher (lower) values of $\mu$ ($P_{0}, P_{\text{FD}}$) result in a higher SEE. Moreover, it is observed that a marginal gain with the application of an FD jammer is obtained for a high $\mu$, and a small $P_{\text{FD}}$ conditions. This is expected, since a small (large) value of $\mu$ ($P_{\text{FD}}$) results in a bigger waste of power when {using an FD jamming strategy}.     

%




\vspace{-0mm}
\subsubsection{Secure bidirectional communication}
\input{./Sections/sim_SUAIP_BD}
In Fig.~\ref{PerfCSI_BD_Pmax} a system with a bidirectional secure communication between Alice and Bob is studied. In particular, a joint FD operation at Alice and Bob is considered which enables jamming and communication simultaneously at both directions. \MinorRevOmid{Two scenarios are considered regarding the decoding capability at Eve: \emph{i}) Eve treats interference from the non-intended information path as noise~(corresponding to $\rho=1$), and \emph{ii}) Eve is capable of decoding, and hence reducing, the received signal from the non-intended information link~(corresponding to $\rho=0$).} Moreover, a setup with HD Bob and HD Alice is also evaluated, where time-division-duplexing (TDD) or frequency-division-duplexing (FDD) is employed in order to facilitate a bidirectional communication. 

It is observed that the resulting SEE increases with $P_{\text{max}}$, however, saturates for high values of maximum transmit power. Moreover, it is observed that a joint FD operation is capable of enhancing the system SEE, with a considerable margin, in the studied bidirectional setup. This is since, due to the coexistence of both communication directions on the same channel the jamming power is re-used for both communication directions, leading to a higher SEE compared to the HD setup. Moreover, the Eve's decoding capability is further decreased in the FD setup considering the scenario (\emph{i}), due to the existence of two information links at the same channel.

\vspace{-0mm}
\subsubsection{FD-enabled jamming with statistical CSI}
\input{./Sections/sim_SSUM}
In Fig.~\ref{fig_SSUM_2} the cumulative distribution function (CDF) of the resulting SEE is evaluated via the application of SSSLM algorithm \revOmid{on $100$ problem instances\footnote{\revOmid{Each problem instance includes a realization of $\ma{H}_{ab},\ma{H}_{bb}$}.}}, where only a statistical CSI is available for the channels to Eve. We choose $|\mathbb{G}_C| = 100$ in the construction of SAA, see Section~\ref{sec_SSUM}, in order to limit the required computational effort. The CDF of the resulting SEE is then evaluated via the utilization of $10000$ channel realizations for each problem instance, following the statistical distribution defined in the beginning of the current section and choosing $\ma{D}_{\mathcal{X}}$ as a matrix of all-$1$ elements. \\
In Fig.~\ref{fig_SSUM_2}~(a) the performance of the SSSLM algorithm with the consideration of the statistical CSI, is compared to the case where the SUIAP algorithm is applied directly on the channel estimate matrices $\tilde{\ma{H}}_{\mathcal{X}}$. It is observed that a significant gain is obtained by taking into account the full channel statistics, however, at the expense of a higher computational complexity. Moreover, the superior SEE performance of the SEE specific design, compared to the secrecy rate maximizing designs is observable.

In Figs.~\ref{fig_SSUM_2}~(b)-(d) the CDF of the resulting SEE is evaluated for different levels of thermal noise ($\sigma_{\text{n}}^2$), hardware inaccuracy ($\kappa$), and the PA efficiency ($\mu$). Similar to the observed trends for the scenario where exact CSI is available, a marginal gain is observed in the resulting SEE with the application of an optimized FD jamming strategy. In particular, the gain of the FD-enabled system is improved for a system with a high SNR, i.e., a high transmit power budget or a low noise level, and as hardware accuracy increases. 

%% file: Sections/sim_alg_analysis.tex
 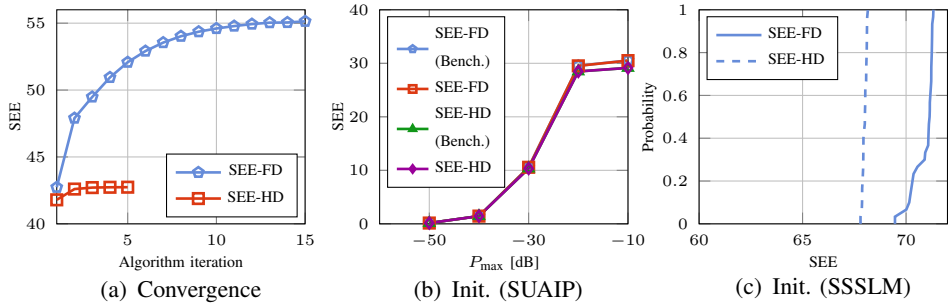
\begin{figure*}[t]  
\centering
\hspace{0.0cm}  \subfigure[Convergence]{		\begin{tikzpicture}[baseline,trim axis left,trim axis right]
	\begin{axis}[myStandard,
	width=.2\textwidth,
	legend style={font=\tiny},
	legend style={cells={align=left}},
	legend entries ={SEE-FD, SEE-HD},
	xmin=1, 
	xmax=15,
	ymin=40,
	ymax=56,
	xlabel={Algorithm iteration},
	ylabel={SEE },
	xlabel near ticks,
	ylabel near ticks,
	xlabel style={at={(0.5,-0.1)}},
	/pgf/number format/1000 sep={},
	clip=false,
	legend image post style={xscale=1},
	/tikz/plot label/.style={black, anchor=west},
	legend pos=south east,
	font=\tiny]
	\addplot table[x index=0,y index=1] {data/convergence_fd_sfp.dat};
	\addplot table[x index=0,y index=1] {data/convergence_hd_sfp.dat};
	\end{axis}
	\end{tikzpicture} } \hspace{0.4cm}
\subfigure[Init. (SUAIP)]{ 			\begin{tikzpicture}[baseline,trim axis left,trim axis right]
		\begin{axis}[myStandard,
			width=.2\textwidth,
			legend style={font=\tiny},
			legend style={cells={align=left}},   
			legend cell align=left,			
			legend entries ={SEE-FD\\(Bench.), SEE-FD, SEE-HD\\(Bench.), SEE-HD},
			xmin=-60, 
			xmax=-10,
			ymin=0,
			ymax=40,
			xtick={-50,-30,-10},			
			xlabel={$P_{\text{max}}$ [dB]},
			ylabel={SEE },
			xlabel near ticks,
			ylabel near ticks,
			xlabel style={at={(0.5,-0.1)}},
			/pgf/number format/1000 sep={},
			clip=false,
			legend image post style={xscale=0.7, yscale=0.7},
			legend pos=north west,
			font=\tiny]
			\addplot table[x index=0,y index=2] {data/bench_data.dat};
			\addplot table[x index=0,y index=1] {data/bench_data.dat};
			\addplot table[x index=0,y index=4] {data/bench_data.dat};
			\addplot table[x index=0,y index=3] {data/bench_data.dat};
	\end{axis}
	\end{tikzpicture}}  \hspace{0.65cm} \subfigure[Init. (SSSLM)]{\begin{tikzpicture}[baseline,trim axis left,trim axis right]
	\begin{axis}[myStandardCDF,
	width=.2\textwidth,
	legend style={font=\tiny},
	legend style={cells={align=left}},
	legend cell align=left,			
	legend entries ={SEE-FD, SEE-HD},
	xmin=60, 
	xmax=72,
	ymin=0,
	ymax=1,
	ylabel={Probability},
	xlabel={SEE},
	xlabel near ticks,
	ylabel near ticks,
	/pgf/number format/1000 sep={},
	xlabel style={at={(0.5,-0.1)}},
	clip=false,
	legend image post style={xscale=1},
	/tikz/plot label/.style={black, anchor=west},
	legend pos=north west,
	font=\tiny]
	\addplot table[x index=0,y index=1] {data/eval_ssum.dat};
	\addplot table[x index=2,y index=3] {data/eval_ssum.dat};
	\end{axis}
	\end{tikzpicture}}   
\caption{Numerical algorithm analysis in terms of the average convergence behavior, initialization, and computational complexity.}  \vspace{0mm} \label{fig_alg_analysis}
\end{figure*}

%% file: Sections/sim_SUAIP.tex
\begin{figure*}[t]  
\centering
\hspace{0.0cm} \subfigure[\tiny{Impact of thermal noise}]{ 			\begin{tikzpicture}[baseline,trim axis left,trim axis right]
			\begin{axis}[myStandard,
								legend cell align=left,			
				width=.2\textwidth,
				legend style={font=\tiny},
				legend style={cells={align=left}},
				legend entries ={SEE-FD, SEE-HD, CS-FD, CS-HD},
				xmin=-50, 
				xmax=-10,
				ymin=0,
				ymax=250,
				xlabel={$\sigma_{\text{n}}^2$ [dB]},
				ylabel={SEE},
				xlabel near ticks,
				ylabel near ticks,
				xlabel style={at={(0.5,-0.15)}},
				/pgf/number format/1000 sep={},
				clip=false,
				legend image post style={xscale=1},
				/tikz/plot label/.style={black, anchor=west}, 
				legend pos=north east,font=\footnotesize]
				\addplot table[x index=0,y index=1] {data/nvar_data.dat};
				\addplot table[x index=0,y index=2] {data/nvar_data.dat};
				\addplot table[x index=0,y index=3] {data/nvar_data.dat};
				\addplot table[x index=0,y index=4] {data/nvar_data.dat};
			\end{axis}
	\end{tikzpicture}}  \hspace{0.3cm} 
	\subfigure[\tiny{Max. Tx. power}]{		\begin{tikzpicture}[baseline,trim axis left,trim axis right]
			\begin{axis}[myStandard,
								legend cell align=left,			
				width=.2\textwidth,
				legend style={font=\tiny},
				legend style={cells={align=left}},
				legend entries ={SEE-FD, SEE-HD, CS-FD, CS-HD},
				xmin=-50, 
				xmax=60,
				ymin=0,
				ymax=100,
				xtick={-50,-25,0,25,50},			
				xlabel={$P_{\text{max}}$ [dB]},
				xlabel near ticks,
				ylabel near ticks,
				xlabel style={at={(0.5,-0.15)}},
				/pgf/number format/1000 sep={},
				clip=false,
				legend image post style={xscale=1},
				/tikz/plot label/.style={black, anchor=west}, 
				legend style={at={(0.95,0.225)},anchor=south east},
				font=\footnotesize]
				\addplot table[x index=0,y index=1] {data/pmax_data.dat};
				\addplot table[x index=0,y index=2] {data/pmax_data.dat};
				\addplot table[x index=0,y index=3] {data/pmax_data.dat};
				\addplot table[x index=0,y index=4] {data/pmax_data.dat};
		\end{axis}
	\end{tikzpicture} } \hspace{0.3cm} \subfigure[\tiny{Impact of hardware accuracy}]{	\begin{tikzpicture}[baseline,trim axis left,trim axis right]
			\begin{axis}[myStandard,
				legend cell align=left,			
				width=.2\textwidth,
				legend style={font=\tiny},
				legend style={cells={align=left}},
				legend entries ={SEE-FD, SEE-HD, CS-FD, CS-HD},
				xmin=-60, 
				xmax=0,
				ymin=0,
				ymax=90,
				xlabel={$\kappa$~[dB]},
				xlabel near ticks,
				ylabel near ticks,
				xlabel style={at={(0.5,-0.15)}},
				/pgf/number format/1000 sep={},
				clip=false,
				legend image post style={xscale=1},
				/tikz/plot label/.style={black, anchor=west}, 
				legend style={at={(0.04,0.25)},anchor=south west},
				font=\footnotesize]
				\addplot table[x index=0,y index=1] {data/kappa_beta_data.dat};
				\addplot table[x index=0,y index=2] {data/kappa_beta_data.dat};
				\addplot table[x index=0,y index=3] {data/kappa_beta_data.dat};
				\addplot table[x index=0,y index=4] {data/kappa_beta_data.dat};
			\end{axis}
	\end{tikzpicture}}  \hspace{0.3cm} 
	\subfigure[\tiny{Impact of Eve's position}]{\begin{tikzpicture}[baseline,trim axis left,trim axis right]
			\begin{axis}[myStandard,
								legend cell align=left,			
				width=.2\textwidth,
				legend style={font=\tiny},
				legend style={cells={align=left}},
				legend entries ={SEE-FD, SEE-HD, CS-FD, CS-HD},
				xmin=-20, 
				xmax=100,
				ymin=0,
				ymax=5,
				xlabel={$d_E$},
				xlabel near ticks,
				ylabel near ticks,
				xlabel style={at={(0.5,-0.15)}},
				/pgf/number format/1000 sep={},
				clip=false,
				legend image post style={xscale=1},
				/tikz/plot label/.style={black, anchor=west}, 
				legend pos=north west,font=\footnotesize]
				\addplot table[x index=0,y index=1] {data/pos_data.dat};
				\addplot table[x index=0,y index=2] {data/pos_data.dat};
				\addplot table[x index=0,y index=3] {data/pos_data.dat};
				\addplot table[x index=0,y index=4] {data/pos_data.dat};
			\end{axis}
	\end{tikzpicture}}  \hspace{0.3cm} 
	\subfigure[\tiny{No. of antennas at Eve}]{	\begin{tikzpicture}[baseline,trim axis left,trim axis right]
			\begin{axis}[myStandard,
								legend cell align=left,			
				width=.2\textwidth,
				legend style={font=\tiny},
				legend style={cells={align=left}},
				legend entries ={SEE-FD, SEE-HD, CS-FD, CS-HD},
				xmin=2, 
				xmax=8,
				ymin=0,
				ymax=200,
				xlabel={$M_E$},
				ylabel={SEE},
				xlabel near ticks,
				ylabel near ticks,
				xlabel style={at={(0.5,-0.15)}},
				/pgf/number format/1000 sep={},
				clip=false,
				legend image post style={xscale=1},
				/tikz/plot label/.style={black, anchor=west}, 
				legend pos=north east,font=\footnotesize]
				\addplot table[x index=0,y index=1] {data/Ne_data.dat};
				\addplot table[x index=0,y index=2] {data/Ne_data.dat};
				\addplot table[x index=0,y index=3] {data/Ne_data.dat};
				\addplot table[x index=0,y index=4] {data/Ne_data.dat};
			\end{axis}
	\end{tikzpicture} }  \hspace{0.3cm} \subfigure[\tiny{PA efficiency}]{ 			\begin{tikzpicture}[baseline,trim axis left,trim axis right]
			\begin{axis}[myStandard,
								legend cell align=left,			
				width=.2\textwidth,
				legend style={font=\tiny},
				legend style={cells={align=left}},
				legend entries ={SEE-FD, SEE-HD, CS-FD, CS-HD},
				xmin=0, 
				xmax=1,
				ymin=0,
				ymax=175,
				xlabel={$\mu$},
				xlabel near ticks,
				ylabel near ticks,
				xlabel style={at={(0.5,-0.15)}},
				/pgf/number format/1000 sep={},
				clip=false,
				legend image post style={xscale=1},
				/tikz/plot label/.style={black, anchor=west}, 
				legend pos=north west,font=\footnotesize]
				\addplot table[x index=0,y index=1] {data/alpha_data.dat};
				\addplot table[x index=0,y index=2] {data/alpha_data.dat};
				\addplot table[x index=0,y index=3] {data/alpha_data.dat};
				\addplot table[x index=0,y index=4] {data/alpha_data.dat};
			\end{axis}
	\end{tikzpicture} } \hspace{0.3cm} \subfigure[\tiny{Zero-state power consumption}]{\begin{tikzpicture}[baseline,trim axis left,trim axis right]
			\begin{axis}[myStandard,
								legend cell align=left,			
				width=.2\textwidth,
				legend style={font=\tiny},
				legend style={cells={align=left}},
				legend entries ={SEE-FD, SEE-HD, CS-FD, CS-HD},
				xmin=-20, 
				xmax=0,
				ymin=0,
				ymax=90,
				xlabel={$P_0$ [dB]},
				xlabel near ticks,
				ylabel near ticks,
				xlabel style={at={(0.5,-0.15)}},
				/pgf/number format/1000 sep={},
				clip=false,
				legend image post style={xscale=1},
				/tikz/plot label/.style={black, anchor=west}, 
				legend pos=north east,font=\footnotesize]
				\addplot table[x index=0,y index=1] {data/Pc_data.dat};
				\addplot table[x index=0,y index=2] {data/Pc_data.dat};
				\addplot table[x index=0,y index=3] {data/Pc_data.dat};
				\addplot table[x index=0,y index=4] {data/Pc_data.dat};
		\end{axis}
	\end{tikzpicture} } \hspace{0.3cm} \subfigure[\tiny{FD mode power consumption}]{		\begin{tikzpicture}[baseline,trim axis left,trim axis right]
			\begin{axis}[myStandard2,
				width=.2\textwidth,
				legend style={font=\tiny},
				legend style={cells={align=center}},
				legend entries ={
					$\sigma_{\text{n}}^2 = -50$ dB,
					$\sigma_{\text{n}}^2 = -50$ dB,
					$\sigma_{\text{n}}^2 = -40$ dB,
					$\sigma_{\text{n}}^2 = -40$ dB,
				},
				xmin=-40, 
				xmax=0,
				ymin=0,
				ymax=700,
				xlabel={$P_{\text{FD}}$ [dB]},
				xlabel near ticks,
				ylabel near ticks,
				xlabel style={at={(0.5,-0.15)}},
				/pgf/number format/1000 sep={},
				clip=false,
				legend image post style={xscale=1},
				/tikz/plot label/.style={black, anchor=west}, 
				legend pos=north east,
				font=\footnotesize]
				\addplot table[x index=0,y index=1] {data/nvar_pfd_data.dat};
				\addplot table[x index=0,y index=2] {data/nvar_pfd_data.dat};
				\addplot table[x index=0,y index=5] {data/nvar_pfd_data.dat};
				\addplot table[x index=0,y index=6] {data/nvar_pfd_data.dat};
			\end{axis}
	\end{tikzpicture}  } 
\caption{SEE performance of the secure communication system with exact CSI, via the utilization of SUAIP algorithm. } \vspace{0mm} \label{fig_SUAIP}
\end{figure*}
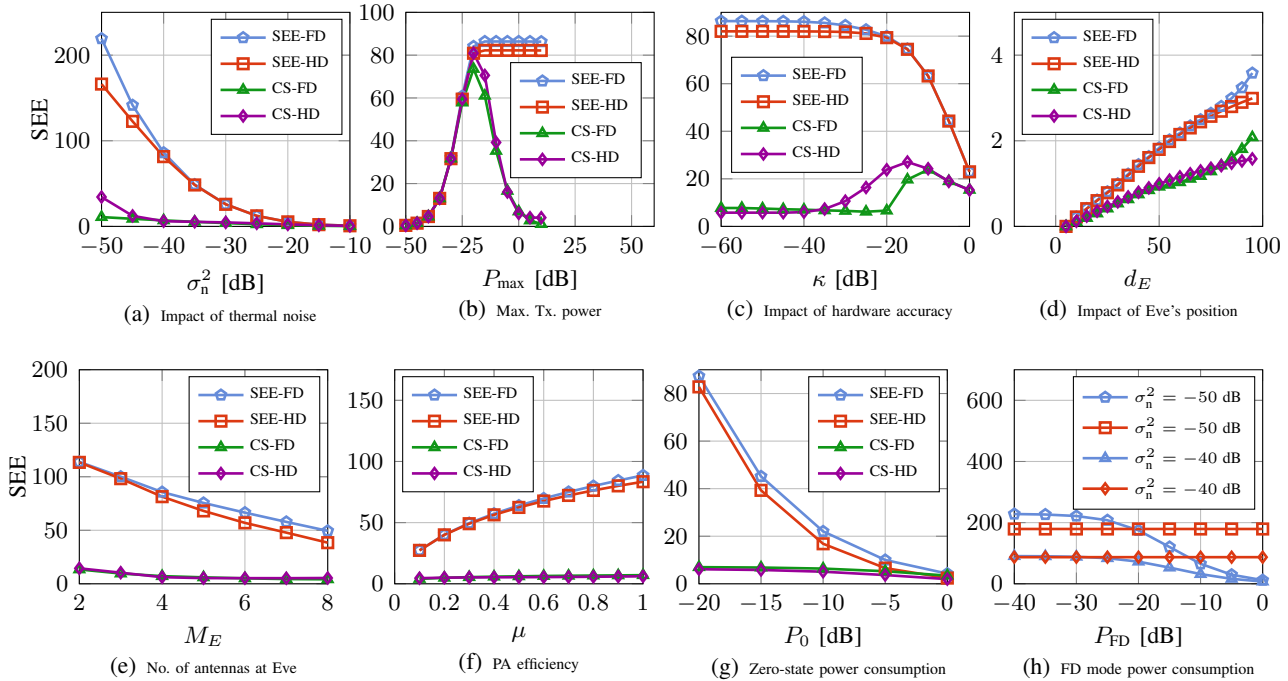 

%% file: Sections/sim_SUAIP_BD.tex
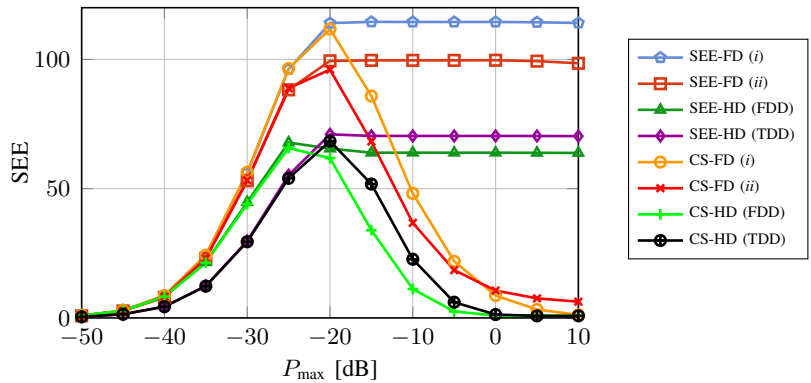
\begin{figure}
\centering
	\begin{tikzpicture}[baseline,trim axis left,trim axis right]
			\begin{axis}[myStandard,
				axis line style=semithick,
				legend cell align=left,						
				width=.4\textwidth,
				height=.25\textwidth,
				legend style={font=\tiny},
				legend style={cells={align=left}},  
				legend entries ={
					SEE-FD (\emph{i}),
					SEE-FD (\emph{ii}),
					SEE-HD (FDD),
					SEE-HD (TDD),
					CS-FD (\emph{i}),
					CS-FD (\emph{ii}),
					CS-HD (FDD),
					CS-HD (TDD),},
				xmin=-50, 
				xmax=10,
				ymin=0,
				ymax=120,
				xtick={0,-50,-40,-30,-20,-10,10},			
				xlabel={$P_{\text{max}}$~[dB]},
				ylabel={SEE},
				xlabel near ticks,
				ylabel near ticks,
				xlabel style={at={(0.5,-0.1)}},
				/pgf/number format/1000 sep={},
				clip=false,
				legend image post style={xscale=1},
				/tikz/plot label/.style={black, anchor=west},
				legend style={at={(1.1,0.9)}, anchor=north west},
				font=\footnotesize]
				\addplot table[x index=0,y index=1] {data/bidirectional_pmax_data.dat};
				\addplot table[x index=0,y index=2] {data/bidirectional_pmax_data.dat};
				\addplot table[x index=0,y index=3] {data/bidirectional_pmax_data.dat};
				\addplot table[x index=0,y index=4] {data/bidirectional_pmax_data.dat};
				\addplot table[x index=0,y index=5] {data/bidirectional_pmax_data.dat};
				\addplot table[x index=0,y index=6] {data/bidirectional_pmax_data.dat};
				\addplot table[x index=0,y index=7] {data/bidirectional_pmax_data.dat};
				\addplot table[x index=0,y index=8] {data/bidirectional_pmax_data.dat};	
			\end{axis}
	\end{tikzpicture}
\caption{SEE performance of FD, FDD and TDD for secure bidirectional communications.}  \vspace{0mm}
\label{PerfCSI_BD_Pmax}
\end{figure}

%% file: Sections/sim_SSUM.tex
 
\begin{figure*}[t]  
\centering
\subfigure[Performance comparison SSSLM~vs.~SUIAP]{\includegraphics[scale=0.65]{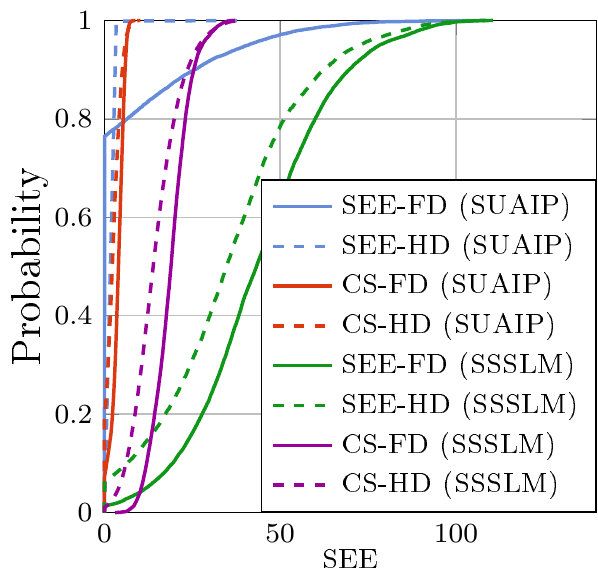}}  
\subfigure[Impact of thermal noise]{\includegraphics[scale=0.65]{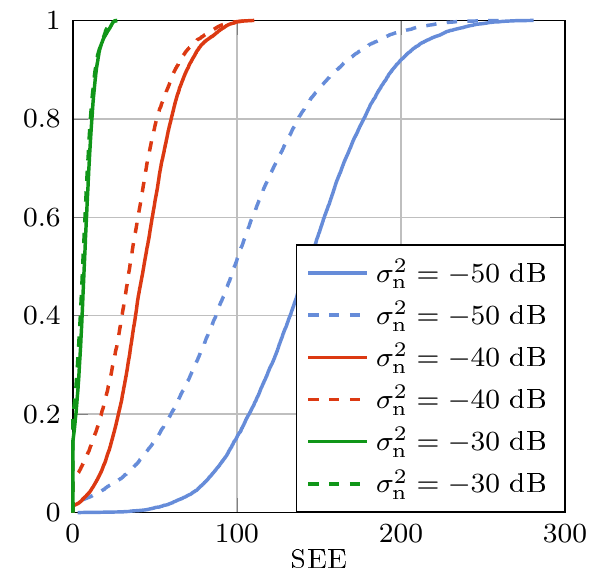}} 
\subfigure[Impact of PA efficiency]{\includegraphics[scale=0.65]{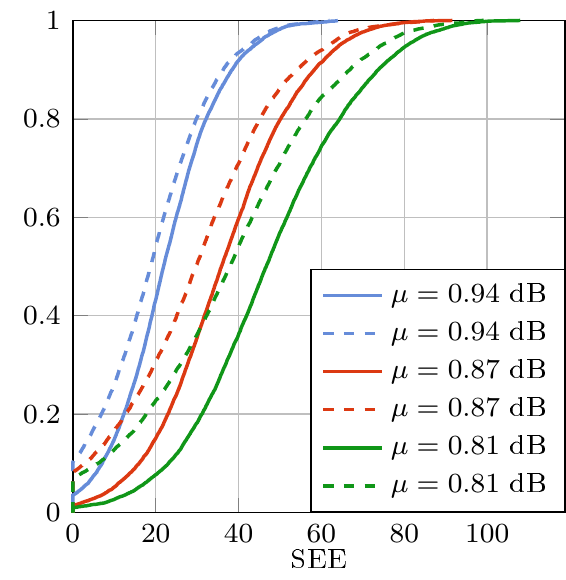}} 
\subfigure[Impact of hardware accuracy]{\includegraphics[scale=0.65]{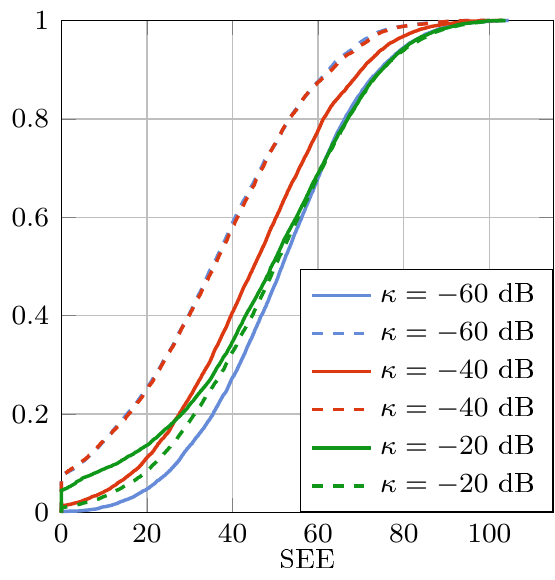}}
\caption{CDF of the resulting SEE under various system conditions utilizing the SSSLM algorithm. The solid (dashed) lines represent the performance of the FD (HD) setups in (a)-(d). }  \label{fig_SSUM_2}
\end{figure*}

%% file: Sections/sec_conclusion.tex
The utilization of FD jamming transceivers \revOmid{is widely known to enhance} the secrecy capacity of wireless communication systems by transmitting AN while exchanging information. However, this results in a higher power consumption in the system due to \emph{i)} the degrading impact of residual self-interference, \emph{ii)} the implementation of \revPeterNew{an} SIC scheme at the FD transceiver, as well as \emph{iii)} the power consumed for the transmission of AN. In this work, \revOmid{we have observed} that the application of FD transceivers \revOmid{results} {only in} a marginal gain in terms of secrecy energy efficiency, for a wide range of system conditions. However, \revOmid{we have shown} that the aforementioned SEE gain becomes significant for a system with a small distance between the FD node \revPeterNew{and the} eavesdropper or a system operating in high SNR regimes, under the condition that the self-interference is efficiently mitigated. Moreover, a promising SEE gain \revPeterNew{has been observed for a bidirectional FD communication setup}, where jamming power \revOmid{is reused} for both directions. \revOmid{We have observed} that for almost all system conditions, the application of \revPeterNew{an SEE-aware design is essential to increase the system's SEE as compared to the available designs which target the maximization of secrecy capacity.}

%% file: Sections/Appendix.tex
\revOmid{\section{Dinkelbach's algorithm} \label{appendix_dinkelbach}
Let $f:\mathbb{R}^n \rightarrow \mathbb{R}$ and $g:\mathbb{R}^n \rightarrow \mathbb{R}^+$ be, respectively, a concave differentiable and a convex differentiable function. Moreover, let $\mathcal{X}$ be a convex compact set in $\mathbb{R}^n$. Then, the optimization problem 
\begin{align}                                   
\underset{ \ma{x} }{\text{max}} \;\; {f(\ma{x})}/{g(\ma{x})}\;\; \text{s.t.}\;\; \ma{x} \in \mathcal{X}  \label{P_CCFP} 
\end{align} 
represents the class of concave over convex fractional programs, see \cite{zappone2015energy,rev_1_1, rev_1_3} for a wide range of applications and related methods .

\begin{lemma} \label{Dinkelbach_aux}
\cite[Section~2]{dinkelbach1967nonlinear} Consider the real-valued auxiliary function 
\begin{align}    \label{dinkelbach_auxilliary_func}                                   
\gamma (\lambda) := \underset{ \ma{x} \in \mathcal{X} }{\text{max}} f(\ma{x}) - \lambda g(\ma{x}).
\end{align} 
Then, $\gamma(\lambda)$ is strictly monotonically decreasing and convex over $\lambda$. Moreover, $\ma{x}^\star \in \mathcal{X}$ is a globally optimum solution to (\ref{P_CCFP}) iff $\ma{x}^\star =  \underset{ \ma{x} \in \mathcal{X} }{\text{arg max}} f(\ma{x}) - \lambda^\star g(\ma{x})$ with $\lambda^\star$ as the unique zero of $\gamma(\lambda)$.     
\end{lemma}
The purpose of the Dinkelbach's algorithm \cite{dinkelbach1967nonlinear} is to obtain the unique zero of $\gamma(\lambda)$, and thereby the global optimum of the fractional problem (\ref{P_CCFP}). This is implemented via iteratively evaluating and updating $\gamma(\lambda)$ in the decreasing direction, where $\gamma(\lambda)$ can be evaluated as a standard convex problem, see Algorithm~\ref{alg_Dinkelbach} for a detailed procedure.
\vspace{-3mm}\begin{algorithm}[H] 
{\tiny{ \revOmid{\begin{algorithmic}[1] 
\State{$ \lambda = 0, \epsilon > 0, \eta > \epsilon;$}   \Comment{initialization}
\Repeat  
\State{$ \ma{x}^\star \leftarrow \underset{ \ma{x} \in \mathcal{X} }{\text{arg max}} f(\ma{x}) - \lambda^\star g(\ma{x}) $} 
\State{$ \eta \leftarrow  f(\ma{x}^\star) - \lambda g(\ma{x}^\star) $} 
\State{$ \lambda^\star \leftarrow  f(\ma{x}^\star)/g(\ma{x}^\star) $} 
\Until{$ \eta \geq \epsilon $}
\State{\Return$ \left\{\lambda^\star, \ma{x}^\star\right\} $}
  \end{algorithmic} }}}
 \caption{\scriptsize{Dinkelbach's Algorithm~\cite{dinkelbach1967nonlinear}} } \label{alg_Dinkelbach}
\end{algorithm}  }

\section{SUIAP initialization}
\subsection{The choice of $\ma{F}$ and $\ma{G}$} \label{appendix_init_spatialadjustent}
As mentioned, the role of $\ma{F}$ ($\ma{G}$) is to direct (suppress) the transmission into the desired (undesired) direction. Hence, for the design of $\ma{Q}_a$, i.e., data transmission from Alice, we choose $\ma{F} \leftarrow \ma{H}_{ab}$, and $\ma{G} \leftarrow \ma{H}_{ae}$. Conversely, for the design of $\ma{W}_a$ we choose $\ma{F} \leftarrow \ma{H}_{ae}$, and $\ma{G} \leftarrow \ma{H}_{ab}$. For the design of $\ma{W}_b$ we set $\ma{F} \leftarrow \ma{H}_{be}$. However, the choice of $\ma{G}$ should include the impact of distortion terms on Bob, reflecting the effect of residual self-interference. The distortion power at Bob can be written as 
\begin{align}                                 
& \text{tr}\Big( \kappa \ma{H}_{bb} \text{diag}\left(\ma{W}_b\right) \ma{H}_{bb}^H \Big) + \text{tr}\Big( \beta \text{diag}\left( \ma{H}_{bb} \ma{W}_b \ma{H}_{bb}^H \right)  \Big) \nonumber \\
& =   \text{tr}\bigg(  \Big( \underbrace{   \kappa \text{diag}\left( \ma{H}_{bb}^H\ma{H}_{bb} \right)  +  \beta \ma{H}_{bb}^H\ma{H}_{bb} }_{\tilde{\ma{H}}_{bb}} \Big) \ma{W}_b \bigg), \nonumber 
\end{align}  
which consequently results in the choice of $\ma{G} \leftarrow \left( {\tilde{\ma{H}}_{bb}} \right)^{\frac{1}{2}}$.
  
\subsection{Power adjustment} \label{appendix_init_powAdj}
Via the utilization of the obtained spatial beams, i.e., normalized covariance matrices, the optimal power adjustment on each transmission is seeked to maximize the resulting SEE. In each case, by fixing the power of the other transmission, the resulting $\text{SEE}_p$ is written as  
\begin{align}  \label{appendix_SEE_pp}                                 
\text{SEE}_p \left( p_\mathcal{X} \right) = \frac{\text{log} \left(  \frac{ \alpha_{11}^{\mathcal{X}} p_\mathcal{X}^2 + \alpha_{12}^{\mathcal{X}} p_\mathcal{X} + \alpha_{13}^{\mathcal{X}} }{  \alpha_{21}^{\mathcal{X}} p_\mathcal{X}^2 + \alpha_{22}^{\mathcal{X}} p_\mathcal{X} + \alpha_{23}^{\mathcal{X}} }  \right) }{ \gamma_1^{\mathcal{X}} p_\mathcal{X} + \gamma_2^{\mathcal{X}}}, 
\end{align}  
where $p_\mathcal{X}$, $\mathcal{X} \in \left\{ \ma{Q}_a, \ma{W}_a, \ma{W}_b \right\}$, represent the power associated with different transmissions. It is observed from (\ref{appendix_SEE_pp}) that $\text{SEE}_p \rightarrow 0$ for $p_\mathcal{X} \rightarrow \infty$ and $\text{SEE}_p \left( 0 \right) $ is a finite and non-negative value. Moreover, $\text{SEE}_p \left( p_\mathcal{X} \right)$ is a continuous and differentiable function in the region $p_\mathcal{X} \in [0, \infty)$. This concludes the location of the optimal $p_\mathcal{X}$ at the problem boundaries, or at the points equalizing the derivative of $\text{SEE}_p \left( p_\mathcal{X} \right)$ to zero, see section \ref{appendix_power_efficient_implementation} for an efficient numerical solution.
\subsection{Efficient implementation} \label{appendix_power_efficient_implementation}
For ease of notation, the objective (\ref{appendix_SEE_pp}) is denoted as $f(p_\mathcal{X})$ and its numerator is defined as $g(p_\mathcal{X})$:
\begin{align}
	\label{eq:obj_f}
	f(p_\mathcal{X}) &{}= \frac{ g(p_\mathcal{X}) }{\gamma_1^{\mathcal{X}} \, p_\mathcal{X} + \gamma_2^{\mathcal{X}}}\\
	\label{eq:obj_num_g}
	g(p_\mathcal{X}) &{}= \log \left( \frac{\alpha_{11}^{\mathcal{X}} \, p_\mathcal{X}^2 + \alpha_{12}^{\mathcal{X}} \, p_\mathcal{X} + \alpha_{13}^{\mathcal{X}}}{\alpha_{21}^{\mathcal{X}} \, p_\mathcal{X}^2 + \alpha_{22}^{\mathcal{X}} \, p_\mathcal{X} + \alpha_{23}^{\mathcal{X}}} \right).
\end{align}
The goal is to find the maximum feasible value $\lambda^*$ of the objective function $f(p_\mathcal{X})$ as well as $p_\mathcal{X}^*$ for which $\lambda^*$ is achieved. $\lambda^*$ is defined as follows:
\begin{equation}
	\lambda^* = \max_{0 \leq p_\mathcal{X} \leq c} \frac{ g(p_\mathcal{X}) }{\gamma_1^{\mathcal{X}} \, p_\mathcal{X} + \gamma_2^{\mathcal{X}}}.
\end{equation}
It can be seen that the objective $f(p_\mathcal{X})$ is continuous and differentiable. Because $p_\mathcal{X}$ is bounded on the closed interval $\left[ 0, c \right]$, it follows from the extreme value theorem that the objective has a maximum $\lambda^*$ on this interval. The optimum $p_\mathcal{X}$, i.e. $p_\mathcal{X}^*$, is located either on the boundaries of the closed interval or on a point satisfying
\begin{equation}
	\label{eq:cond_stationary_point}
	\frac{\partial}{\partial p_\mathcal{X}^*} \frac{ g(p_\mathcal{X}^*) }{\gamma_1^{\mathcal{X}} \, p_\mathcal{X}^* + \gamma_2^{\mathcal{X}}} = 0.
\end{equation}

Using the quotient rule, (\ref{eq:cond_stationary_point}) can be rewritten as
\begin{align}
	\label{eq:cond2}
	&\Leftrightarrow \frac{ g'(p_\mathcal{X}^*) \left( \gamma_1^{\mathcal{X}} \, p_\mathcal{X}^* + \gamma_2^{\mathcal{X}} \right) - g(p_\mathcal{X}^*) \, \gamma_1^{\mathcal{X}}}{ \left( \gamma_1^{\mathcal{X}} \, p_\mathcal{X}^* + \gamma_2^{\mathcal{X}} \right)^2 } = 0\\
	&\Leftrightarrow g'(p_\mathcal{X}^*) \left( \gamma_1^{\mathcal{X}} \, p_\mathcal{X}^* + \gamma_2^{\mathcal{X}} \right) - g(p_\mathcal{X}^*) \, \gamma_1^{\mathcal{X}} = 0\\
	\label{eq:optimality_condition}
	&\Leftrightarrow \frac{g'(p_\mathcal{X}^*)}{\gamma_1^{\mathcal{X}}} = \frac{ g(p_\mathcal{X}^*) }{\gamma_1^{\mathcal{X}} \, p_\mathcal{X}^* + \gamma_2^{\mathcal{X}}} = \lambda,
\end{align}
where $g'(p_\mathcal{X})$ denotes the derivative of $g(p_\mathcal{X})$ with respect to $p_\mathcal{X}$, given by
\begin{equation}
	\label{eq:g_diff_def}
	g'(p_\mathcal{X}) {}= \frac{2 \, \alpha_{11}^{\mathcal{X}} \, p_\mathcal{X} + \alpha_{12}^{\mathcal{X}} }{ \alpha_{11}^{\mathcal{X}} \, p_\mathcal{X}^2 + \alpha_{12}^{\mathcal{X}} \, p_\mathcal{X} + \alpha_{13}^{\mathcal{X}}} - \frac{2 \, \alpha_{21}^{\mathcal{X}} \, p_\mathcal{X} + \alpha_{22}^{\mathcal{X}} }{ \alpha_{21}^{\mathcal{X}} \, p_\mathcal{X}^2 + \alpha_{22}^{\mathcal{X}}  \, p_\mathcal{X} + \alpha_{23}^{\mathcal{X}}}.
\end{equation}
Our goal is to convert the maximization problem into a simpler feasibility problem: For a given value of the objective function, denoted as $\lambda$, check if a feasible $p_\mathcal{X}$ exists. Therefore, equation (\ref{eq:optimality_condition}) is rewritten to:
\begin{equation}
	\label{eq:optimality_condition_rewritten}
	\lambda \, \gamma_1^{\mathcal{X}} - \frac{2 \, \alpha_{11}^{\mathcal{X}} \, p_\mathcal{X} + \alpha_{12}^{\mathcal{X}} }{ \alpha_{11}^{\mathcal{X}} \, p_\mathcal{X}^2 + \alpha_{12}^{\mathcal{X}} \, p_\mathcal{X} + \alpha_{13}^{\mathcal{X}}} + \frac{2 \, \alpha_{21}^{\mathcal{X}} \, p_\mathcal{X} + \alpha_{22}^{\mathcal{X}} }{ \alpha_{21}^{\mathcal{X}} \, p_\mathcal{X}^2 + \alpha_{22}^{\mathcal{X}}  \, p_\mathcal{X} + \alpha_{23}^{\mathcal{X}}} = 0.
\end{equation}
Equation (\ref{eq:optimality_condition_rewritten}) is a fourth-order polynomial and hence, it can be written as
\begin{equation}
	\label{eq:polynomial}
	c_4 \, p_\mathcal{X}^4 + c_3 \, p_\mathcal{X}^3 + c_2 \, p_\mathcal{X}^2 + c_1 \, p_\mathcal{X} + c_0 = 0,
\end{equation}
with
\begin{align*}
	c_4 &{}= \alpha_{11}^{\mathcal{X}} \, \alpha_{21}^{\mathcal{X}} \, \gamma_1^{\mathcal{X}} \lambda\\
	c_3 &{}= \left( \alpha_{12}^{\mathcal{X}} \, \alpha_{21}^{\mathcal{X}} + \alpha_{11}^{\mathcal{X}} \, \alpha_{22}^{\mathcal{X}} \right) \gamma_1^{\mathcal{X}} \lambda\\
	c_2 &{}= \left( \alpha_{13}^{\mathcal{X}} \, \alpha_{21}^{\mathcal{X}} + \alpha_{12}^{\mathcal{X}} \, \alpha_{22}^{\mathcal{X}} + \alpha_{11}^{\mathcal{X}} \, \alpha_{23}^{\mathcal{X}} \right) \gamma_1^{\mathcal{X}} \lambda - \alpha_{11}^{\mathcal{X}} \, \alpha_{22}^{\mathcal{X}} + \alpha_{12}^{\mathcal{X}} \, \alpha_{21}^{\mathcal{X}}\\
	c_1 &{}= \left( \alpha_{13}^{\mathcal{X}} \, \alpha_{22}^{\mathcal{X}} + \alpha_{12}^{\mathcal{X}} \, \alpha_{23}^{\mathcal{X}} \right) \gamma_1^{\mathcal{X}} \lambda + 2 \, \alpha_{13}^{\mathcal{X}} \, \alpha_{21}^{\mathcal{X}} - 2 \, \alpha_{11}^{\mathcal{X}} \, \alpha_{23}^{\mathcal{X}}\\
	c_0 &{}= \alpha_{13}^{\mathcal{X}} \, \alpha_{23}^{\mathcal{X}} \, \gamma_1^{\mathcal{X}} \lambda - \alpha_{12}^{\mathcal{X}} \, \alpha_{23}^{\mathcal{X}} + \alpha_{13}^{\mathcal{X}} \, \alpha_{22}^{\mathcal{X}}.
\end{align*}
Let the roots of (\ref{eq:polynomial}) be denoted as $p_\mathcal{X}^i$, $i \in \{1,\ldots,4\}$. If there is any real $p_\mathcal{X}^i \in \left[ 0, c \right]$ for which
\begin{equation}
	\label{eq:check_if_feasible}
	f(p_\mathcal{X}^i) \geq \lambda
\end{equation}
holds, then $\lambda$ is feasible. As a result, it is possible to construct the following bi-section algorithm to find the maximum $\lambda$, denoted as $\lambda^*$.

Firstly, a closed interval for $\lambda^*$ is defined, i.e. $\lambda^* \in \left[ \lambda_{\min} , \lambda_{\max} \right]$, where $\lambda_{\min}$ can usually be defined as
\begin{equation}
	\label{eq:def_lambda_min}
	\lambda_{\min} = \min \left\lbrace f(0), f(c) \right\rbrace.
\end{equation}
Moreover, $\lambda_{\max}$ can be chosen as an upper bound on $\lambda$:
\begin{equation}
	\label{eq:lambda_max}
	\lambda_{\max} = \frac{\log \left(  \max \left\lbrace \alpha_{11}^{\mathcal{X}}/\alpha_{21}^{\mathcal{X}}  , \alpha_{12}^{\mathcal{X}}/\alpha_{22}^{\mathcal{X}},  \alpha_{13}^{\mathcal{X}}/\alpha_{23}^{\mathcal{X}} \right\rbrace \right)  }{ \gamma_2^{\mathcal{X}}}.
\end{equation}
The algorithm finds $\lambda^*$ up to some tolerance $\epsilon > 0$. In the first iteration it is verified if $\lambda^1 = \frac{\lambda_{\max} - \lambda_{\min}}{2}$ is feasible. Therefore, the roots $p_\mathcal{X}^i$, $i \in \{1,\ldots,4\}$ of (\ref{eq:polynomial}) are calculated for $\lambda = \lambda^1$. If any $p_\mathcal{X}^i$ is real, lies in the interval $\left[ 0, c \right]$ and $f(p_\mathcal{X}^i) \geq \lambda^1$ holds, then $\lambda^1$ is feasible and the procedure is repeated for $\lambda^2 = \frac{\lambda_{\max} - \lambda^1}{2}$. Otherwise $\lambda^1$ is infeasible and the procedure is repeated for $\lambda^2 = \frac{\lambda^1 - \lambda_{\min}}{2}$. By construction this algorithm numerically approximates the maximum feasible objective value $\lambda^*$ to arbitrary precision $\epsilon$. This algorithm is formally given by algorithm \ref{alg:bi-section_maximization}.
\begin{algorithm}[H]
 	{\scriptsize{\begin{algorithmic}[1]
 		\State{\textbf{Input: }$\lambda_{\min}$, $\lambda_{\max}$, $c$;}
  		\State{$\epsilon > 0$, $\ell = 0$, $a = \lambda_{\min}$, $b = \lambda_{\max}$;}
  		\Repeat
  		\State{$\ell = \ell+1$;}
  		\State{$\text{isFeasible} =$ \textbf{false};}
  		\State{$\lambda^\ell = \frac{a+b}{2}$;}
  		\State{Calculate roots $p_\mathcal{X}^{i,\ell}$, $i \in \{1,\ldots,4\}$ of (\ref{eq:polynomial}) for $\lambda^\ell$.}
  		\ForAll{$i \in \left\lbrace 1, \ldots, 4 \right\rbrace$}
  			\If{ $\mathrm{Im}\left\lbrace p_\mathcal{X}^{i,\ell} \right\rbrace = 0$ \textbf{and} $p_\mathcal{X}^{i,\ell} \in \left[0, c\right]$ \textbf{and} $f\left(p_\mathcal{X}^{i,\ell}\right) \geq \lambda^\ell$}
  				\State{$\text{isFeasible} =$ \textbf{true};}
  				\State{$p_\mathcal{X} = p_\mathcal{X}^{i,\ell}$;}
  				\State{\textbf{break};}
  			\EndIf
  		\EndFor
  		\If{ $\text{isFeasible}$ }
  			\State{$a = \lambda^\ell$;}
  		\Else
  			\State{$b = \lambda^\ell$;}
  		\EndIf
		\Until{$\text{isFeasible}$ \textbf{and} $\frac{b-a}{2} < \epsilon$}
		\State{\Return{$p_\mathcal{X}^* = p_\mathcal{X}$, $\lambda^* = \lambda^\ell$;}}
		\end{algorithmic}} }
	\caption{{Bi-Section Power Allocation} } \label{alg:bi-section_maximization}
\end{algorithm} 

\MinorRR{
\vspace{-8mm}\section{Proof to Lemma~\ref{lemma_BD_Positive_C}} \label{appendix_lemma_BD_Positive_C}
Let $\mathcal{A}_1:=\left( \ma{Q}_a^{\star},\ma{Q}_b^{\star},\ma{W}_a^{\star}, \ma{W}_b^{\star} \right)$ be a KKT solution for (36). Moreover, let ${\tilde{C}^{\text{BD}}_{ab}} \left( \mathcal{A}_1 \right) < 0$, without loss of generality\footnote{The case assuming ${\tilde{C}^{\text{BD}}_{ba}} \left( \mathcal{A}_1 \right) < 0$ can be argued similarly. Moreover, since KKT conditions are also necessary conditions for any globally optimum solution to (36), due to the differentiable objective with linear constraints, the given proof in this part also subsumes the case when $\mathcal{A}_1$ is a globally optimum solution.}. The proof is obtained via contradiction as follows.
The Lagrangian function, corresponding to the problem (36) is formulated as
\begin{align}    \label{}                                   
\mathcal{L} \Big( \ma{Q}_a,\ma{Q}_b, & \ma{W}_a, \ma{W}_b , \overbar{\ma{Q}_a}, \overbar{\ma{Q}_b},\overbar{\ma{W}_a},\overbar{\ma{W}_b}, \tau_a, \tau_b \Big) = - \text{SEE}_p^{\text{BD}}  + \tau_a \left( P^{\text{BD}}_{A} - P_{A,\text{max}}\right) + \tau_b \left( P^{\text{BD}}_{B} - P_{B,\text{max}}\right) \nonumber \\ 
& - \text{tr}\left( \overbar{\ma{Q}_a} \ma{Q}_a \right) - \text{tr}\left( \overbar{\ma{Q}_b} \ma{Q}_b\right) - \text{tr}\left( \overbar{\ma{W}_a} \ma{W}_a\right) - \text{tr}\left( \overbar{\ma{W}_b} \ma{W}_b\right), \nonumber
\end{align}
where $\tau_a,\tau_b \geq 0$ are slack variables associated with the power constraints, whereas $\overbar{\ma{Q}_a},\overbar{\ma{Q}_b},$ $\overbar{\ma{W}_a},\overbar{\ma{W}_b} \in \mathcal{H}$ are slack variables for dualizing the semidefinite constraints. Since $\mathcal{A}_1$ is a KKT solution, the directional derivative of the Lagrangian function must vanish for any direction at the point $\mathcal{A}_1$. In order to utilize this property, we observe the behavior of the Lagrangian function when $\mathcal{A}_1$ moves over the directions $d \left(\ma{X} \right)$, such that 
{{ \begin{align}                                    
d \left(\ma{X} \right) := \left( - \ma{U}_q \ma{X} \ma{U}_q^H, \ma{0}_{N_B \times N_B},  \ma{U}_q \ma{X} \ma{U}_q^H, \ma{0}_{N_B \times N_B} \right),\;\; \forall \ma{X} \succeq {0}.                                            
\end{align} }}
\hspace{-2mm}In the above definition, $\ma{Q}_a^{\star} = \ma{U}_q\ma{\Lambda}_q\ma{U}_q^H$, with $\ma{\Lambda}_q \in \compl^{r_q \times r_q}$, is the economy-size singular value decomposition\footnote{This choice ensures that the signal space of the movement $\ma{U}_q \ma{X} \ma{U}_q^H$ remains within the space of $\ma{Q}_a^{\star}$, where $r_q$ represents the rank of $\ma{Q}_a^{\star}$. When $\ma{Q}_a^{\star}$ is not rank-deficient, $\ma{U}_q$ can be simply chosen as an identity matrix.}. Now, let $\nabla_{d} f \left(x\right)$ represent the directional derivative of a function $f$ at point $x$ and for the direction $d$. Then, we have
{{ \begin{align}                                    
 \nabla_{d \left(\ma{X} \right)} \mathcal{L}\left(\mathcal{A}_1\right)  & \overset{(a)}{=} {0} , \;\;\;\;\;\; \forall  \ma{X} \succeq {0}, \nonumber \\
{\Rightarrow} \;\;\; \nabla_{d \left(\ma{X} \right)} - \left(\tilde{C}^{\text{BD}}_{ab} \left(\mathcal{A}_1\right) + \tilde{C}^{\text{BD}}_{ba} \left(\mathcal{A}_1\right) \right) & \overset{(b)}{\geq} {0}, \;\;\;\;\;\; \forall  \ma{X} \succeq {0},  \nonumber  \\
{\Rightarrow} \;\;\; \nabla_{d \left(\ma{X} \right)} \text{log}\left| \ma{\Sigma}_{b}^{\text{BD}} \left(\mathcal{A}_1 \right) \right| - \nabla_{d \left(\ma{X} \right)}\text{log}\left| \ma{\Sigma}_{e-a}^{\text{BD}} \left(\mathcal{A}_1 \right)  \right| \;\;\;\;\;\;\;\;\;\;\;\;\;\;\;\;\;\;\;\;\;\;\;\;\;\;\;\;\;\;\;\;\;\;\;\;\;\;\;\;\;\;\;\;\;\;\;\; &                      
 \nonumber \\ - \underbrace{ \nabla_{d \left(\ma{X} \right)} \left( \text{log}\left| \ma{H}_{ab} \ma{Q}_a \ma{H}_{ab}^H + \ma{\Sigma}_{b}^{\text{BD}} \right| -  \text{log}\left| \ma{H}_{ae} \ma{Q}_a \ma{H}_{ae}^H + \ma{\Sigma}_{e-a}^{\text{BD}} \right| \right) }_{= 0} & \overset{(c)}{\geq} {0}, \;\;\;\;\; \forall  \ma{X} \succeq {0},  \nonumber  \\
{\Rightarrow} \;\;\; \nabla_{d \left(\ma{X} \right)}   \text{log}\left| \ma{\Sigma}_{b}^{\text{BD}} \left(\mathcal{A}_1 \right) \right| - \nabla_{d \left(\ma{X} \right)} \text{log}\left| \ma{\Sigma}_{e-a}^{\text{BD}} \left(\mathcal{A}_1 \right) \right| & \; {\geq} \; 0, \;\;\;\;\;\; \forall  \ma{X} \succeq {0}, \nonumber  \\
{\Rightarrow} \;\;\; \text{tr}\left( \ma{U}_q^H \ma{H}_{ab}^H \left( \ma{\Sigma}_b^{\text{BD}} \left(\mathcal{A}_1 \right) \right)^{-1} \ma{H}_{ab} \ma{U}_q  \ma{X}\right) - \text{tr}\left( \ma{U}_q^H \ma{H}_{ae}^H \left( \ma{\Sigma}_{e-a}^{\text{BD}} \left(\mathcal{A}_1 \right)\right)^{-1} \ma{H}_{ae} \ma{U}_q  \ma{X}\right) &\overset{(d)}{\geq} 0, \;\;\;\;\; \forall  \ma{X} \succeq {0}, \nonumber \\
{\Rightarrow} \;\;\; \ma{U}_q^H \ma{H}_{ab}^H \left( \ma{\Sigma}_b^{\text{BD}} \left(\mathcal{A}_1 \right)\right)^{-1} \ma{H}_{ab} \ma{U}_q \;\; {\succeq} \ma{U}_q^H \ma{H}_{ae}^H \hspace{-1mm} \left( \ma{\Sigma}_{e-a}^{\text{BD}} \left(\mathcal{A}_1 \right) \right)^{\hspace{-1mm}-1} \hspace{-1mm}\ma{H}_{ae} \ma{U}_q & .  \label{Lemma_SemiDefiniteInequality}  
\end{align} }} 
\hspace{-2mm}In the above statements, ($a$) holds as $\mathcal{A}_1$ satisfies the KKT conditions. ($b$) follows from the fact that $\nabla_{d \left(\ma{X} \right)} P_{\text{tot}} = 0$, together with the complementary slackness condition, leading to $\nabla_{d \left(\ma{X} \right)} \text{tr}\left( \overbar{\ma{Q}_a} \ma{Q}_a\right) = 0$, and $\nabla_{d \left(\ma{X} \right)} \text{tr}\left( \overbar{\ma{W}_a} \ma{W}_a\right) \geq 0$. ($c$) is obtained by recalling (\ref{C_BD_ab}),~(\ref{C_BD_ba}) and observing the fact that $ \nabla_{d \left(\ma{X} \right)} \tilde{C}^{\text{BD}}_{ba} \left(\mathcal{A}_1\right) =0$, for the case that $\rho = 1$, and $\nabla_{d \left(\ma{X} \right)} \tilde{C}^{\text{BD}}_{ba} \left(\mathcal{A}_1\right) \geq 0$, when $\rho = 0$. Finally, ($d$) follows from the known identities $\text{tr}\left(\ma{A}\ma{B}\right)=\text{tr}\left(\ma{B}\ma{A}\right)$ and $ \partial \text{log}\left| \ma{A} \right| = \text{tr}\left(\ma{A}^{-1} \partial \ma{A} \right)$. 

Now, recalling the initial assumption ${\tilde{C}^{\text{BD}}_{ab}} \left( \mathcal{A}_1 \right) <  0$ yields
{{\begin{align}    \label{lemma_last_contradiction}                                   
\frac{\left| \ma{\Sigma}_{e-a}^{\text{BD}} \left(\mathcal{A}_1 \right)\right|}{\left|  \ma{\Sigma}_b^{\text{BD}} \left(\mathcal{A}_1 \right)  \right|} \overset{(e)}{<}  \frac{\left| \ma{\Sigma}_{e-a}^{\text{BD}} \left(\mathcal{A}_1 \right) + \ma{H}_{ae} \ma{Q}_a^{\star} \ma{H}_{ae}^H \right|}{\left|  \ma{\Sigma}_b^{\text{BD}} \left(\mathcal{A}_1 \right) + \ma{H}_{ab} \ma{Q}_a^{\star} \ma{H}_{ab}^H \right|} & \nonumber \\
& \hspace{-62mm} \overset{(f)}{=}  \frac{\left| \ma{\Sigma}_{e-a}^{\text{BD}} \left(\mathcal{A}_1 \right)\right|}{\left|  \ma{\Sigma}_b^{\text{BD}} \left(\mathcal{A}_1 \right)  \right|} \times \underbrace{ \left( \frac{\left| \ma{I}+ \ma{U}_q^H \ma{H}_{ae}^H\left( \ma{\Sigma}_{e-a}^{\text{BD}} \left(\mathcal{A}_1 \right) \right)^{-1} \ma{H}_{ae} \ma{U}_q \ma{\Lambda}_q  \right|}{\left|  \ma{I} +  \ma{U}_q^H \ma{H}_{ab}^H \left( \ma{\Sigma}_b^{\text{BD}} \left(\mathcal{A}_1 \right)\right)^{-1}\ma{H}_{ab} \ma{U}_q \ma{\Lambda}_q  \right|} \right) }_{\leq 1} \leq \frac{\left| \ma{\Sigma}_{e-a}^{\text{BD}} \left(\mathcal{A}_1 \right)\right|}{\left|  \ma{\Sigma}_b^{\text{BD}} \left(\mathcal{A}_1 \right) \right|},
\end{align} }}
\hspace{-2mm}which leads to a contradiction. In the above inequalities, ($e$) is obtained by incorporating (\ref{C_BD_ab}) in the inequality ${\tilde{C}^{\text{BD}}_{ab}} \left( \mathcal{A}_1 \right) <  0$ and ($f$) is obtained by recalling (\ref{Lemma_SemiDefiniteInequality}) and employing the matrix identity $\left|\ma{I} + \ma{A}\ma{B}\right| = \left|\ma{I} + \ma{B}\ma{A} \right|$, and the fact that $\left|\ma{I} + \ma{A}\right| \geq \left|\ma{I} + \ma{B} \right|$ for any $\ma{A} \succeq \ma{B}$.} 
%
%
%
%
\section{Proof to Lemma~\ref{lemma_SSUM_Convergence}}  \label{appendix_Lemma_SSUM_convergence}
\subsection{Proof of tightness:}
Tightness is obtained by observing the equivalence
{{\begin{subequations} \label{appendix_SSUM_lemma_tigtness}
\begin{align}                                   
& |\mathbb{F}_{{C}} |  P_{\text{tot}}\left( \mathbb{Q}^\star \right) \text{SAA} \left( \mathbb{Q}^\star \right) =   \sum_{i \in \mathbb{F}_C}  \{ \tilde{C}_{\text{s},i} \left( \mathbb{Q}^\star \right)  \}^+   
\overset{(g)}{=} \sum_{i \in \mathbb{G}_{C_1}} \{ \tilde{C}_{\text{s},i} \left( \mathbb{Q}^\star \right) \}^+ + \sum_{i \in \mathbb{G}_{C_2^+}} \tilde{C}_{\text{s},i} \left( \mathbb{Q}^\star \right)    \label{appendix_SSUM_lemma_tigtness_a} \\
& \overset{(h)}{=} \sum_{i \in \mathbb{G}_{C_1}} \{ \hat{C}_{\text{s},i} \left( \mathbb{Q}^\star,  \mathbb{Q}^\star \right) \}^+ + \sum_{i \in \mathbb{G}_{C_2^+}} \hat{C}_{\text{s},i} \left( \mathbb{Q}^\star ,  \mathbb{Q}^\star \right) 
=  |\mathbb{G}_{{C}} |  P_{\text{tot}}\left( \mathbb{Q}^\star \right) \text{SAA}_{LB} \left(\mathbb{Q}^\star,  \mathbb{Q}^\star \right), \nonumber 
\end{align} 
\end{subequations} }}
where ($g$) is obtained by applying (\ref{SSUM_set_def_C_2_+}), and ($h$) from $\tilde{C}_{\text{s},i} \left( \mathbb{Q}^\star\right) = \hat{C}_{\text{s},i} \left( \mathbb{Q}^\star,  \mathbb{Q}^\star \right)$, see (\ref{eq_SEE_max_SEE_Taylor}).  \vspace{-4mm} 
\vspace{-5mm}\subsection{Proof of equal directional derivative:}
Let $C_{\text{s},i} := \{ \tilde{C}_{\text{s},i} \}^+$. The directional derivative of $\text{SAA}$ at $\mathbb{Q}^\star$ is then expressed as 
{{\begin{subequations}\label{directional_derivative_lemma_ssum}
\begin{align}                                   
&  P_{\text{tot}}\left( \mathbb{Q}^\star \right) {\MinorRR{ \nabla_{d}}} \text{SAA} \left( \mathbb{Q}^\star \right)  \nonumber \\
& {=}  \Big( \sum_{i \in \mathbb{G}_{C_1}} {\MinorRR{ \nabla_{d}}} C_{\text{s},i} \left( \mathbb{Q}^\star \right)    + \sum_{i \in \mathbb{G}_{C_2^+}} {\MinorRR{ \nabla_{d}}} \tilde{C}_{\text{s},i} \left( \mathbb{Q}^\star \right) \Big)/|\mathbb{G}_{{C}} |  - {\MinorRR{ \nabla_{d}}} P_{\text{tot}} \left( \mathbb{Q}^\star \right) \text{SAA}\left( \mathbb{Q}^\star \right)  \label{directional_derivative_lemma_ssum_a} \\
& =  \Big( \sum_{i \in \mathbb{G}_{C_1^{(d)}}} {\MinorRR{ \nabla_{d}}} \hat{C}_{\text{s},i} \left(\mathbb{Q}^\star, \mathbb{Q}^\star \right)    + \sum_{i \in \mathbb{G}_{C_2^+}} {\MinorRR{ \nabla_{d}}} \hat{C}_{\text{s},i} \left(\mathbb{Q}^\star, \mathbb{Q}^\star  \right) \Big)/|\mathbb{G}_{{C}} |   - {\MinorRR{ \nabla_{d}}} P_{\text{tot}} \left( \mathbb{Q}^\star \right) \text{SAA}_{LB}\left( \mathbb{Q}^\star, \mathbb{Q}^\star \right) \label{directional_derivative_lemma_ssum_b}   \\
& =  P_{\text{tot}}\left( \mathbb{Q}^\star \right) {\MinorRR{ \nabla_{d}}} \text{SAA}_{LB}\left( \mathbb{Q}^\star, \mathbb{Q}^\star \right),  \label{directional_derivative_lemma_ssum_c}
\end{align} 
\end{subequations} } }
where the set $\mathbb{G}_{C_1^{(d)}}$ is defined as 
\begin{align}                                 
\mathbb{G}_{C_1^{(d)}} := \left\{  \forall i \;\; \vert \;\; i \in {\mathbb{G}_{{C_1}}} \;\;  \text{and}  \;\; {\MinorRR{ \nabla_{d}}} C_{\text{s},i} \left( \mathbb{Q}^\star \right) \neq 0  \right\}.
\end{align} 
In the above arguments, (\ref{directional_derivative_lemma_ssum_a}) is obtained by recalling (\ref{Eq_SSUM_SAA}), and the fact that ${C}_{\text{s},i} \left( \mathbb{Q}^\star\right)$ is positive and differentiable for any $i \in \mathbb{G}_{C_2^+}$. The identity (\ref{directional_derivative_lemma_ssum_b}) is obtained by considering the possible situations for ${\MinorRR{ \nabla_{d}}} C_{\text{s},i} \left( \mathbb{Q}^\star \right)$: 
\begin{itemize}
\item $\tilde{C}_{\text{s},i} \left( \mathbb{Q}^\star \right) < 0 .$ Then, $C_{\text{s},i}$ is differentiable and ${\MinorRR{ \nabla_{d}}} C_{\text{s},i} \left( \mathbb{Q}^\star \right) = 0$ for any direction $d$.
\item $\tilde{C}_{\text{s},i} \left( \mathbb{Q}^\star \right) > 0 .$ Then, $C_{\text{s},i}$ is differentiable and ${\MinorRR{ \nabla_{d}}} C_{\text{s},i} \left( \mathbb{Q}^\star \right) = {\MinorRR{ \nabla_{d}}} \hat{C}_{\text{s},i} \left( \mathbb{Q}^\star, \mathbb{Q}^\star \right)$, $\forall d$.
\item $\tilde{C}_{\text{s},i} \left( \mathbb{Q}^\star \right) = 0 $ and ${\MinorRR{ \nabla_{d}}} \tilde{C}_{\text{s},i} \left( \mathbb{Q}^\star \right) > 0$. Then, $C_{\text{s},i}$ is not differentiable and ${\MinorRR{ \nabla_{d}}} C_{\text{s},i} \left( \mathbb{Q}^\star \right) = {\MinorRR{ \nabla_{d}}} \hat{C}_{\text{s},i} \left( \mathbb{Q}^\star, \mathbb{Q}^\star \right)$. 
\item $\tilde{C}_{\text{s},i} \left( \mathbb{Q}^\star \right) = 0 $ and ${\MinorRR{ \nabla_{d}}} \tilde{C}_{\text{s},i} \left( \mathbb{Q}^\star \right) \leq 0$. Then, $C_{\text{s},i}$ is not differentiable and ${\MinorRR{ \nabla_{d}}} C_{\text{s},i} \left( \mathbb{Q}^\star \right) = 0$. 
\end{itemize}
Finally, the identity (\ref{directional_derivative_lemma_ssum_c}) is obtained by recalling (\ref{SSUM_SAA_LB}), and the tightness property from (\ref{appendix_SSUM_lemma_tigtness}).